\documentclass{article}

\usepackage[margin=1in]{geometry}
\usepackage{amsmath,amssymb,amsthm}
\usepackage{tikz}
\usetikzlibrary{arrows}
\usetikzlibrary{decorations,intersections}
\usepackage{kbordermatrix}
\usepackage{float}
\usepackage{appendix}
\usepackage{url}
\usepackage{xspace}
\usepackage{nicefrac}
\usepackage[inline]{enumitem}
\usepackage[mathscr]{euscript}

\newtheorem{theorem}{Theorem}
\newtheorem{lemma}[theorem]{Lemma}

\theoremstyle{remark}

\definecolor{grey}{rgb}{0.5,0.5,0.5}
\definecolor{lightgrey}{rgb}{1,0.8,0.8} 

\floatstyle{plain}
\restylefloat{figure}

\definecolor{grey}{rgb}{0.5,0.5,0.5}

\newcommand{\vthsp}{\hspace*{0.25pt}}
\newcommand{\thsp}{\hspace*{0.5pt}}

\newcommand{\RG}{\mathcal{RG}}
\newcommand{\squ}{\fbox{1}}
\newcommand{\one}{\textbf{1}}
\newcommand{\forba}{\raisebox{-1pt}{\,\tikz{\draw[line width=1pt,scale=0.17] (0,0)--(0,1.5)--(1,1.5);}\thsp}\xspace}
\newcommand{\forbb}{\raisebox{-0.25pt}{\,\tikz{\draw[line width=1pt,scale=0.17] (0,0)--(1,0)--(1,1.5);}\thsp}\xspace}
\newcommand{\forbc}{\raisebox{-1pt}{\tikz{\,\draw[line width=1pt,scale=0.17] (0,0)--(1,1.5);}\thsp}\xspace}
\newcommand{\gammafree}{\forba\hspace{-4pt}-\thsp free\xspace}

\newcommand{\nbh}{\mathcal{N}}
\newcommand{\dist}{\textup{dist}}

\newcommand{\per}{\textrm{per}}

\newcommand{\CE}{\mathcal{E}}
\newcommand{\CG}{\mathcal{G}}

\newcommand{\Om}{\Omega}

\newcommand{\N}{\mathbb{N}}

\newcommand{\DGH}{Diaconis, Graham and Holmes\xspace}
\newcommand{\tmix}{T_{\textup{mix}}}

\newcommand{\rset}{\mathbb{R}}
\def\focusA{\alpha}
\def\focusB{\beta}
\def\tprob{\mathcal{P}}
\let\epsilon=\varepsilon
\def\Vbar{\overline{V_k}}

\begin{document}

\title{On the switch Markov chain for perfect matchings\thanks{The work described here was partially
supported by EPSRC grants  ``Computational Counting'' (refs.~EP/I011935/1 and EP/I012087/1)
and ``Randomized algorithms for computer networks'' (ref.~EP/M004953/1)}}
\author{Martin Dyer\thanks{
        School of Computing, University of Leeds, Leeds, LS2~9JT, UK.}
        \and Mark Jerrum\thanks{
        School of Mathematical Sciences, Queen Mary University of
        London, Mile End Road, London, E1~4NS, UK.}
        \and Haiko M\"{u}ller\footnotemark[2]}

\date{January 26, 2017}%

\maketitle

\begin{abstract}
We study a simple Markov chain, the switch chain,
on the set of all perfect matchings in a bipartite graph.
This Markov chain was proposed by Diaconis, Graham and Holmes as a possible approach
to a sampling problem arising in Statistics.  We ask:~for which hereditary classes of graphs is
the Markov chain ergodic and for which is it rapidly mixing?  We provide a precise answer to
the ergodicity question and close bounds on the mixing question.  We show for the first
time that the mixing time of the switch chain is polynomial in the case of monotone
graphs, a class that includes examples of interest in the statistical setting.
\end{abstract}

\renewcommand{\arraystretch}{1.1}

\tikzset{every picture/.style={line width=0.8pt}}
\tikzset{empty/.style={rectangle,draw=none,fill=none}}
\tikzset{cnode/.style={circle,draw,inner sep=0pt,fill=black,minimum size=1.2mm}}
\tikzset{snode/.style={circle,draw,inner sep=0pt,fill=black,minimum size=0.9mm}}
\tikzset{rnode/.style={circle,draw,inner sep=0pt,fill=white,minimum size=1.2mm}}
\tikzset{xnode/.style={rectangle,draw,inner sep=0pt,rotate=45,fill=black,minimum size=1.75mm}}
\tikzset{ynode/.style={rectangle,draw,inner sep=0pt,rotate=45,fill=white,minimum size=1.75mm}}
\tikzset{tnode/.style={rectangle,draw,inner sep=0pt,fill=black,minimum size=1.5mm}}
\tikzset{lnode/.style={empty,anchor=base}}
\tikzset{circlenode/.style={circle,draw,inner sep=2pt,minimum size=5mm}}
\tikzset{middlearrow/.style={decoration={markings,mark= at position 0.5 with {\arrow{#1}},},postaction={decorate}}}

\section{Introduction}\label{sec:intro}

Counting perfect matchings in a bipartite graph is an important computational problem: the permanent of a 0-1 matrix. This, and the computationally equivalent problem of sampling matchings uniformly at random, has applications in Statistics, Statistical Physics and other areas. \DGH\cite{DiGrHo01} considered applications of the 0-1 permanent to problems in Statistics, in particular where the 0-1 matrix has recognisable structure, which they called \emph{truncated} or \emph{interval-restricted}.

The truncated 0-1 matrices are those for which the columns can be permuted to give the \emph{consecutive~1's} property on rows. That is, no two 1's in any row are separated by one or more 0's. \DGH~\cite{DiGrHo01} considered ``one-sided'' truncation, where the consecutive 1's appear at the left of each row, and ``two-sided'' truncation, where the consecutive 1s can appear anywhere in a row. For two-sided truncation, they considered two special cases. In the first, both the rows and columns can be permuted so that they have the consecutive 1's property. In the second, the rows and columns can be permuted so that the consecutive 1's have a ``staircase'' presentation. This is the \emph{monotone} case, which is of particular interest in certain statistical applications~\cite{EfrPet99}.

\DGH~\cite{DiGrHo01} proposed a Markov chain for sampling perfect matchings in a bipartite graph, which we call the \emph{switch chain}. They showed ergodicity of the chain for the truncated matrices considered in~\cite{DiGrHo01}, and conjectured that it would converge rapidly.  Computing the 0-1 permanent is a well-solved problem from a theoretical viewpoint. It is \#P-complete to compute exactly~\cite{Valian79a}, but there is a polynomial time approximation algorithm~\cite{JeSiVi04}. However, the switch chain gives a simpler and more practical algorithm than that of~\cite{JeSiVi04}, making it worthy of consideration. Hence \DGH's conjecture was subsequently studied in the PhD theses of Matthews~\cite{Matthe08} and Blumberg~\cite{Blumbe12}. We will discuss their results below.

A 0-1 matrix is equivalently the biadjacency matrix of a bipartite graph, and we will study the graphs corresponding to
the matrices considered by \DGH~\cite{DiGrHo01}. We show that their questions are, in fact, questions about counting perfect matchings in \emph{hereditary graph classes}. Hereditary classes are of central interest in modern graph theory. In fact, this paper can be viewed as a first attempt to bring together two previously distinct streams of research in algorithms: structural graph theory and the complexity of counting and sampling.

We identify the largest hereditary graph class in which the switch chain is ergodic: \emph{chordal bipartite} graphs. The classes considered in~\cite{DiGrHo01} form an ascending sequence within this class. We examine the mixing time behaviour of the switch chain for graphs from these classes, extending the work of~\cite{DiGrHo01},~\cite{Matthe08} and~\cite{Blumbe12}. We also discuss the difficulty of computing the permanent \emph{exactly} in these classes, where Valiant's result~\cite{Valian79a} does not necessarily imply \#P-completeness.

In particular, we show for the first time that the mixing time of the switch chain is polynomial for monotone graphs. This is proved by a novel application of a simple combinatorial lemma: the solution to the so-called \emph{mountain climbing problem}~\cite{GoPaYa89,Keleti93,Pak10,Tucker95,West01}. Though this lemma is well known, there appears to be no worst-case analysis of this problem in the literature. Therefore we provide such an analysis as an appendix.

After this paper was written, we learned that Bhatnagar, Randall, Vazirani and 	Vigoda~\cite{BhRaVV06} had used a similar approach for a rather different problem. They had analysed the Jerrum-Sinclair Markov chain~\cite{JerSin89} for generating random
bichromatic matchings in a graph that has its edges partitioned into two colour classes.

For the necessary background information on Markov chains, see \emph{inter alia}~\cite{AldFil02,Jerrum03,LePeWi06}. For the relevant graph-theoretic background, see~\cite{BrLeSp99,West01,Spinra03,Golumb04}.\vspace{-3.8mm}

\subsection{Notation and definitions}\label{sec:notation}\vspace{-1.2mm}
Let $\N=\{1,2,\ldots\}$ denote the natural numbers. If $n\in\N$, let $[n]=\{1,2,\ldots,n\}$ and,
if $n_1,n_2\in\N$, let $[n_1,n_2]=\{n_1,n_1+1,\ldots,n_2\}$.

We will also use the notation $[n]'=\{1',2',\ldots,n'\}$ and $[n_1,n_2]'=\{n'_1,(n_1+1)',\ldots,n'_2\}$. Here the prime serves only to distinguish $i$ from $i'$. Ordering and arithmetic for $[n]'$ elements follows that for $[n]$. Thus, for example, $1'<2'$ and $1'+2'=3'$.

A graph $G=(V,E)$ is \emph{bipartite} if its vertex set $V=[m]\cup[n]'$ and there is no (undirected) edge $(v,w)\in E$ such that $v,w\in[m]$ or $v,w\in[n]'$. Thus $V$ comprises two independent sets $[m]$ and $[n]'$.
Bipartite graphs $G_1=([m]\cup[n]',E_1)$ and $G_2=([m]\cup[n]',E_2)$ are \emph{isomorphic} if there are permutations $\sigma$ of $[m]$ and $\tau$ of $[n]'$ such that $(j,k')\in E_1$ if and only if $(\sigma_{j},\tau_{k'})\in E_2$.

Let $G=([m]\cup[n]',E)$ be a bipartite graph. We consider $[m]$ and $[n]'$ to have the usual linear ordering, and we will abuse notation by denoting these ordered sets simply by $[m]$ and $[n]'$. Then let $A(G)$ denote the $m\times n$ \emph{biadjacency matrix} of $G$, with rows indexed by $[m]$ and columns by $[n]'$, such that $A(i,j')=1$ if $(i,j')\in E$, and $A(i,j')=0$ otherwise. We will use the graph and matrix terminology interchangeably. For example, we refer to
rows and columns of $G$, or edges in $A(G)$.

The neighbourhood  in $G$ of a vertex $v\in [m]\cup [n]'$  will be denoted by $\nbh(v)$. To avoid trivialities, we will assume that $G$ has no isolated vertices, unless explicitly stated otherwise.

A \emph{matching} in a bipartite graph $G=([m]\cup[n]',E)$ is a set of independent edges, that is, no two edges in the set share an endpoint. A \emph{perfect} matching is a set of edges such that every vertex of $G$ lies in exactly one edge. For a bipartite graph $G=([m]\cup[n]',E)$ this requires $m=n$, and $n$ independent edges in $E$. In particular, $G$ can have no isolated vertices. We will call a bipartite graph with $m=n$ \emph{balanced}.

Equivalently, a perfect matching may be viewed as $n$ independent $1$'s in the $n\times n$\, 0-1 matrix $A(G)$.
Thus a perfect matching $M$  has edge set $\{(i,\pi'_i):i\in[n]\}$, where $\pi$ is a permutation of $[n]$. Equivalently, $M$ has edge set $\{(\sigma_j,j'):j\in[n]\}$, where $\sigma$ is a permutation of $[n]$. Note that $\sigma=\pi^{-1}$ as elements of the symmetric group~$S_n$. We may identify the matching~$M$ with the permutations $\pi$ and~$\sigma$. An example is shown in Fig.~\ref{med:fig001} below.
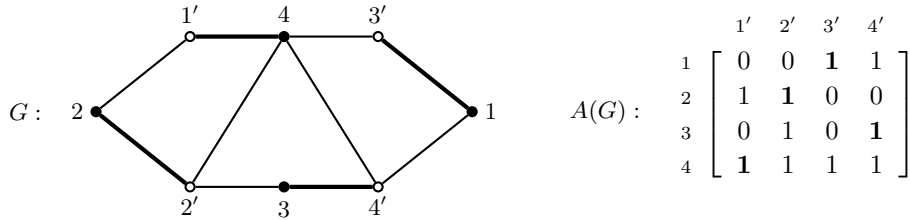
\begin{figure}[h]
\centering{\begin{tikzpicture}[xscale=1.25,yscale=1,font=\small]
\draw (0,0) node[cnode] (2) {} +(-.2,0) node {2}
++(1,1)  node[rnode] (1')  {} +(0,.3) node {$1'$}
++(1,0)  node[cnode] (4)  {} +(0,.3) node {4}
++(1,0)  node[rnode] (3') {} +(0,.3) node  {$3'$}
++(1,-1)  node[cnode] (1)  {} +(.2,0) node  {1}
++(-1,-1)  node[rnode] (4')  {} +(0,-.25) node {$4'$}
++(-1,-0)  node[cnode] (3)  {} +(0,-.3) node {3}
++(-1,0)  node[rnode] (2')  {} +(0,-.25) node {$2'$} ;
\draw (1) -- (3')  (1) --(4') (2)--(1')  (2) -- (2') (3)--(2') (3)--(4')
(4)--(1') (4)--(2') (4)--(3') (4)--(4') ;
\draw[line width=1.6pt] (1) to (3') (2) to (2') (3) to (4') (4) to (1');
\draw (-0.75,0) node[empty] {$G:$} (5.5,0) node[empty] {$A(G):\ \ $} ;
\end{tikzpicture}
\raisebox{14mm}{$\kbordermatrix{%
    & 1' & 2' & 3' & 4'  \\
1   & 0 & 0 & \one & 1  \\
2   & 1 & \one & 0 & 0 \\
3  & 0 & 1 & 0 & \one\\
4  &  \one & 1 & 1 & 1
}$}}\vspace{3mm}
\caption{Bipartite graph with perfect matching $\sigma=(3241)$, $\pi=(4213)$.}\label{med:fig001}
\end{figure}

The total number of perfect matchings in a bipartite graph $G$ is the \emph{permanent}, which we denote by $\per(A)$ when $A=A(G)$.

\subsection{Computing the permanent}\label{sec:permanent}

The permanent has been studied extensively in Combinatorics and Computer Science.
Valiant showed that computing the permanent \emph{exactly} is \#P-complete for a general 0-1 matrix~\cite{Valian79a}. No algorithm running in sub-exponential time is known for the exact evaluation of the permanent of 0-1 matrices.

Jerrum, Sinclair and Vigoda~\cite{JeSiVi04} showed that the 0-1 permanent has a \emph{fully polynomial randomised  approximation scheme} (FPRAS), using an algorithm for randomly sampling perfect matchings. This improved a Markov chain algorithm of Jerrum and Sinclair~\cite{JerSin89}, which was not guaranteed to have polynomial time convergence for all bipartite graphs. The algorithm of~\cite{JeSiVi04} is simple, but involves polynomially many repetitions of a polynomial-length sequence of related Markov chains. The best bound known for the running time of this algorithm is $O(n^7\log^4 n)$, due to Bez\'{akov\'{a}, \v{S}tefankovi\v{c}}, Vazirani and Vigoda~\cite{BeStVV08}.

Jerrum, Valiant and Vazirani~\cite{JeVaVa86} showed that sampling almost uniformly at random  and approximate counting have equivalent computational complexity for many combinatorial problems, including the permanent. Technically, the problem must be \emph{self-reducible}.

From the viewpoint of theoretical Computer Science, these results entirely settle the question of sampling and counting perfect matchings in bipartite graphs.
However, simpler methods have been proposed for special cases of this problem, and here we consider one such proposal.

\subsection{The switch chain}\label{sec:switch}
\DGH~\cite{DiGrHo01}  proposed the following Markov chain for sampling perfect matchings from a balanced bipartite graph $G=([n]\cup[n]',E)$ almost uniformly at random, which we will call the \emph{switch chain}. A transition of the chain will be called a \emph{switch}. \DGH~\cite{DiGrHo01} called this a ``transposition''. The switch chain generalises the transposition chain for generating random permutations.

{\small\textbf{Switch chain}}\vspace{-10pt}
\begin{enumerate}
\setcounter{enumi}{-1}
  \item Let the perfect matching $M_t$ at time $t$ be the  permutation $\pi$ of $[n]$.
  \setcounter{enumi}{0}
  \item Set $t\gets0$, and let $M_0$ be any perfect matching of $G$.
  \item Choose $i,j\in[n]$, uniformly at random,  so $(i,\pi'_i),\,(j,\pi'_j)\in M_t$.\label{chain:step2}
  \item If $i\neq j$ and $(i,\pi'_j),(j,\pi'_i)$ are both in $E$,\\ \hspace*{2cm}set $M_{t+1}\gets M_t\setminus\{(i,\pi'_i),(j,\pi'_j)\}\cup\{(i,\pi'_j),(j,\pi'_i)\}$.
  \item Otherwise, set $M_{t+1}\gets M_t$.
  \item Set $t\gets t+1$. If $t<t_{\max}$, repeat from step~(\ref{chain:step2}). Otherwise, stop.
\end{enumerate}
Note that the switch chain is invariant under isomorphisms of $G$, so properties of the chain can be investigated from the viewpoint of graph theory. An example of a switch is shown below, with the edges $(4,1'),\,(2,2')$ being switched for $(4,2'),\,(2,1')$.
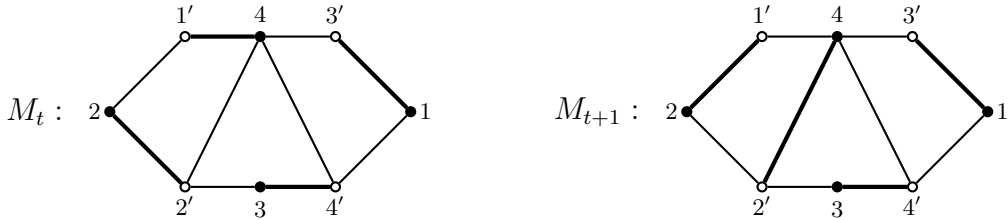
\vspace{1ex}

\begin{figure}[h]
\centering{%
\begin{tikzpicture}[xscale=1,font=\small]
\draw (0,0) node[cnode] (2) {} +(-.2,0) node {2}
++(1,1)  node[rnode] (1')  {} +(0,.3) node {$1'$}
++(1,0)  node[cnode] (4)  {} +(0,.3) node {4}
++(1,0)  node[rnode] (3') {} +(0,.3) node  {$3'$}
++(1,-1)  node[cnode] (1)  {} +(.2,0) node  {1}
++(-1,-1)  node[rnode] (4')  {} +(0,-.25) node {$4'$}
++(-1,-0)  node[cnode] (3)  {} +(0,-.3) node {3}
++(-1,0)  node[rnode] (2')  {} +(0,-.25) node {$2'$} ;
\draw (1) -- (3')  (1) --(4') (2)--(1')  (2) -- (2') (3)--(2') (3)--(4')
(4)--(1') (4)--(2') (4)--(3') (4)--(4') ;
\draw[line width=1.6pt] (1) to (3') (2) to (2') (3) to (4') (4) to (1');
\draw (-1,0) node[empty] {\large $M_t:$} ;
\end{tikzpicture}
\hspace*{0.5in}\begin{tikzpicture}[xscale=1,font=\small]
\draw (0,0) node[cnode] (2) {} +(-.2,0) node {2}
++(1,1)  node[rnode] (1')  {} +(0,.3) node {$1'$}
++(1,0)  node[cnode] (4)  {} +(0,.3) node {4}
++(1,0)  node[rnode] (3') {} +(0,.3) node  {$3'$}
++(1,-1)  node[cnode] (1)  {} +(.2,0) node  {1}
++(-1,-1)  node[rnode] (4')  {} +(0,-.25) node {$4'$}
++(-1,-0)  node[cnode] (3)  {} +(0,-.3) node {3}
++(-1,0)  node[rnode] (2')  {} +(0,-.25) node {$2'$} ;
\draw (1) -- (3')  (1) --(4') (2)--(1')  (2) -- (2') (3)--(2') (3)--(4')
(4)--(1') (4)--(2') (4)--(3') (4)--(4') ;
\draw[line width=1.6pt] (1) to (3') (4) to (2') (3) to (4') (2) to (1');
\draw (-1.2,0) node[empty] {\large $M_{t+1}:$} ;
\end{tikzpicture}
}
\end{figure}

The remainder of the paper is arranged as follows.  Section~\ref{sec:classes} introduces the
hereditary graph classes relevant for our study.  We determine the largest such class for which
the switch chain given above is ergodic, that is, converges eventually to a stationary distribution.
This is the class of chordal bipartite graphs.  Unfortunately, the mixing time of the switch chain on
chordal bipartite graph, though defined, is exponential.
We then examine increasingly restricted graph classes, and determine that the mixing time remains
exponential.  Finally, we arrive at the classes of monotone graphs and chain graphs.
For chain graphs, an explicit formula is known for the number of perfect matchings.
Section~\ref{sec:canonical} contains the main result of the paper, namely that the mixing time of
the switch chain on monotone graphs is polynomial.  Section~\ref{sec:conclusions} briefly
considers possible directions for further work.  The appendix contains a refined analysis of the
``mountain climbing problem'', the solution of which provides a key element in the proof of
polynomial-time mixing for monotone graphs.

\section{Graph classes}\label{sec:classes}
\subsection{Chordal bipartite graphs}\label{sec:chordal}
The first question we might ask about the switch chain is: for which class of graphs is it ergodic?
We wish to have a graph-theoretic answer to this question, so that we can recognise membership of graphs in the class.
Therefore, we restrict attention to \emph{hereditary} graph classes, that is, those for which
all (vertex) induced subgraphs of every graph in the class are also in the class.
Hereditary classes have become central to graph theory, and are most usually characterised by
describing a minimal set of ``excluded  subgraphs'', that is, induced subgraphs which cannot be present. For example, \emph{perfect graphs} are those which exclude all odd-length cycles (\emph{odd holes}) of length at least 5, or their complements (\emph{odd antiholes})~\cite{ChRoST06}. Thus, in particular, the class of perfect graphs contains all bipartite graphs, which exclude all odd holes and antiholes. All the graphs we consider here are bipartite, and hence perfect.

In our application, there is another technical reason for preferring to work with hereditary graph classes: in hereditary classes we have  self-reducibility for most problems in \#P, including the permanent. This property implies equivalence between almost uniform sampling and approximate counting, referred to in Section~\ref{sec:permanent}. See~\cite{JeVaVa86} for details.

The switch chain is \emph{ergodic} on a graph $G=(V,E)$ if the state space of the chain, the set of perfect matchings, is connected by switches. Importantly, we extend this to include graphs which have no perfect matching, where the state space is empty.
Then we will say that a graph $G$ is \emph{hereditarily ergodic} if, for every $U\subseteq V$, the induced subgraph $G[U]$ is ergodic.  A class  of graphs will be called hereditarily ergodic if every graph in the class is hereditarily ergodic.
For a hereditary graph class, this is clearly equivalent to saying that the switch chain is ergodic for every graph in the class.

\DGH~\cite{DiGrHo01} observed that the switch chain is not ergodic for all balanced bipartite graphs.
They gave the example:\vspace{-1ex}
\begin{figure}[h]
\centering{%
\begin{tikzpicture}[font=\small,xscale=1.5,yscale=1]
\draw (0,0) node[cnode] (1) {} +(-.15,0) node {1}
++(0.5,1)  node[rnode] (2') {} +(0,.3) node {$2'$}
++(1,0)  node[cnode] (3) {} +(0,.3) node {3}
++(0.5,-1)  node[rnode] (1') {} +(.15,0) node {$1'$}
++(-0.5,-1)  node[cnode] (2) {} +(0,-.3) node {2}
++(-1,0)  node[rnode] (3') {} +(0,-.3) node {$3'$};
\draw[line width=1.5pt] (1) -- (2') (3)--(1') (3')--(2) ;
\draw[line width=1.5pt,dotted]   (2') --(3)  (2) -- (1') (3')--(1);
\end{tikzpicture}\hspace*{2cm}
\raisebox{14mm}{$\kbordermatrix{%
    & 1' & 2' & 3'  \\
1   & 0 & \one & 1  \\[0.5ex]
2   & 1 & 0 & \one   \\[0.5ex]
3  & \one & 1 & 0 }$}
}
\end{figure}

This graph has two perfect matchings, but the switch chain cannot move between them.
This is because the graph is a chordless $6$-cycle. In fact, this property characterises the
class of graphs for which the switch chain is not ergodic, as we now show.

We say a graph $G$ is \emph{chordal bipartite} if it has no chordless cycle of length other than four.
The class of chordal bipartite graphs is clearly hereditary. Note that the definition of chordal bipartite graphs is an excluded subgraph characterisation. To show that the switch chain is ergodic for this class, we require the following ``excluded submatrix'' characterisation.

A \forba (Gamma) in a 0-1 matrix is an induced submatrix of the form
\[ \forba\;\textup{:} \hspace*{10pt}\begin{bmatrix} 1 & 1\\ 1 & 0 \end{bmatrix}. \]
A matrix is called \gammafree if it has no such induced submatrix. Lubiw~\cite{Lubiw87} gave the following characterisation.
\begin{theorem}[Lubiw]\label{med:thm001}
A bipartite graph is chordal bipartite if and only if it is isomorphic to a graph $G$ such that $A(G)$ is \gammafree.\qed
\end{theorem}
Moreover, Lubiw~\cite{Lubiw87} showed that this characterisation can be used to recognise chordal bipartite graphs in $O(\thsp|E|\thsp\log |E|\thsp)$ time. This was subsequently improved to $O(\thsp|E|\thsp)$ time by Uehara~\cite{Uehara02}. For the switch chain we have:
\begin{lemma}\label{med:lem001}
 The largest hereditary class of bipartite graphs  for which the switch chain is ergodic is the class of chordal bipartite graphs. In this class, if $G=([n]\cup[n]',E)$, the diameter of the chain is at most $n$.
\end{lemma}
\begin{proof}
Clearly any graph with an induced cycle of length greater than 4 cannot be in the class, so we need only show ergodicity for chordal bipartite graphs. The chain is aperiodic, since there is a loop probability at least $1/n$ at each step, from choosing $i=j$ in step~\ref{chain:step2}. Thus we must show that the chain is irreducible. From Theorem~\ref{med:thm001}, we may suppose that $A(G)$ is given with a \gammafree presentation.

Let $\CG=(\Om,\CE)$ be the transition graph of the switch chain, with $\Om$ the set of perfect matchings in $G$, and $\CE$ the set of transitions. We must show that $\CG$ is connected, and has diameter at most $n$. Let $\pi$ and $\sigma$ be any two perfect matchings in $G$, and let $\dist(\pi,\sigma)=|\{i:\pi'_i\neq\sigma'_i\}|$. Note that $\dist(\pi,\sigma)\leq n$,
and $\dist(\pi,\sigma)=0$ implies $\pi=\sigma$.

Let $k$ be the smallest index such that $\pi'_k\neq\sigma'_k$ and, without loss of generality, suppose $\pi'_k>\sigma'_k$. Then there exists $\ell>k$ such that $\pi'_\ell=\sigma'_k$, and hence $\pi'_\ell\neq\sigma'_\ell$. Then we have $(k,\sigma'_k),\,(k,\pi'_k),\,(\ell,\sigma'_k)\in E$, $\ell>k$ and $\pi'_k>\sigma'_k$.
\[ \kbordermatrix{ & \pi'_\ell=\sigma'_k & \ \pi'_k\ \\k & 1 & 1\\ \ell & 1 & ?}\]
The \gammafree property of $A(G)$ then implies $(\ell,\pi'_k)\in E$.
Thus we have $(k,\pi'_k),\,(\ell,\pi'_\ell)\in\pi$ and $(k,\pi'_\ell),\,(\ell,\pi'_k)\in E$. Therefore $\tau\in\Om$ and $(\pi,\tau)\in\CE$, where
\[ \tau= \pi\setminus\{(k,\pi'_k),(\ell,\pi'_\ell)\} \cup  \{(k,\pi'_\ell),(\ell,\pi'_k)\}. \]
Note that $\tau'_i=\pi'_i$ for $i\neq k,\,\ell$. However, $\pi'_k \neq \sigma'_k$, but $\tau'_k=\pi'_\ell=\sigma'_k$. Also $\pi'_\ell \neq \sigma'_\ell$, but $\tau'_\ell = \pi'_k=\sigma'_\ell$ if $\pi'_k=\sigma'_\ell$.
Thus $\dist(\pi,\sigma)-2\leq\dist(\tau,\sigma)\leq\dist(\pi,\sigma) - 1$.
Hence there is a path of at most $n$ edges in $\CG$ between any pair of matchings $\pi,\,\sigma$. Therefore the diameter of $\CG$ is at most~$n$.
\end{proof}
Computing the permanent exactly is known to be \#P-complete for the class of chordal bipartite graphs~\cite{OkUeUn10}, though this result does not extend even to chordal bipartite graphs  of bounded degree. The complexity of exact computation of the permanent is unknown for all the subclasses of chordal bipartite graphs considered below, with the exception of bounded-degree \emph{convex graphs}, which we consider in Section~\ref{sec:convex}, and \emph{chain graphs}, which we examine in Section~\ref{sec:chain}.

\subsection{Convex graphs}\label{sec:convex}
The largest class of graphs considered by \DGH~\cite{DiGrHo01} were those with ``two-sided restrictions''. These are bipartite graphs $G$ for which $A(G)$ has the \emph{consecutive~1's} property. These have been called \emph{convex} graphs in the graph theory literature. They were introduced by Glover~\cite{Glover67}, who gave a simple greedy algorithm for finding a maximum matching in such a graph.  The consecutive~1's property can be recognised in $O(\thsp|E|\thsp)$ time by the well-known algorithm of Booth and Lueker~\cite{BooLue76}.

A bipartite graph is \emph{convex} if it is isomorphic to a graph $G=([m]\cup[n]',E)$ such that $\nbh(i)$ is an interval $[\alpha'_i,\beta'_i\thsp] \subseteq [n]'$ for all $i\in [m]$. Note that this property remains true under an arbitrary permutation of $[m]$. Then
\begin{lemma}\label{med:lem003}
Convex graphs are a proper hereditary subclass of the class of chordal bipartite graphs.
\end{lemma}
\begin{proof}
It is easy to see that the class \textsc{Convex} is hereditary. To see that it is a subclass of chordal bipartite graphs, we  permute rows so that $\beta'_i\leq\beta'_j$ when $i<j$. Now we can show that $A(G)$ is \gammafree. If not, there is a \forba in some rows $i,j$ and columns $k',\ell'$.\vspace{-2ex}
\begin{equation*}
 \kbordermatrix{ & k' & \ell'\\i & 1 & 1\\j & 1 & 0}
\end{equation*}
We have $i<j$ but, since the rows of $A(G)$ have consecutive~1's, $\beta'_i\geq\ell'>\beta'_j$. This contradicts our ordering of the rows of $A(G)$. To see that it is a proper subclass, note that there are at most $n!\binom{n}{2}^n=2^{O(n\log n)}$ labelled convex graphs with $n$ rows and columns, whereas Spinrad~\cite{Spinra95} has shown that there are $2^{\Theta(n\log^2 n)}$ chordal bipartite graphs.
(Spinrad also gives in \cite[Ex.~9.16(c)]{Spinra95} an explicit example
of a graph that is chordal bipartite but not convex.)
\end{proof}
It is possible to give excluded subgraph and excluded submatrix characterisations of convex graphs, but we will not explore this here, since they are not easy to describe, and appear to have little algorithmic application. See~\cite{Tucker72} for details.

Blumberg~\cite{Blumbe12} gave an $r^{O(r)}n^4$ bound on the mixing time of the switch chain for convex graphs with $r=\max_{i\in [n]}\deg(i)$. This is a hereditary subclass of convex graphs, since it is easy to see that graphs with  bounded row- or column-degree form a hereditary subclass of any hereditary class. We will give \emph{exact} algorithms for  counting and sampling in this subclass of convex graphs. First we will show that these graphs also have bounded column degree.
\begin{lemma}\label{med:lem004}
  Let $G=([n] \cup [n]',E)$ be a convex graph  containing a
  perfect matching. Let $r=\max_{i\in [n]}\deg(i)$ and $c=\max_{j\in [n]}\deg(j')$. Then we
  have $c \le 2r-1$.
\end{lemma}
\begin{proof}
 Let $M$ be any perfect matching of $G$. We first permute the rows of $A(G)$ so that $M$ is the diagonal of $A$, i.e. $M \gets \{(i,i') :i \in [n]\}$. To bound $c$, consider any edge $(i,j') \in E$. Since $G$ is
convex, and $(i,i')\in E$, we have $i',j'\in[\alpha'_i,\beta'_i]$ and
so $|i-j| \le r-1$. Hence $i \in [j-r+1,j+r-1]\cap[n]$, and so $\nbh(j')\subseteq [j-r+1,j+r-1]\cap[n]$.   Therefore we have $c \le 2r-1$.
\end{proof}
It is known that there is an exact algorithm for computing the permanent which is linear in $n$ for all graphs of bounded \emph{treewidth}~\cite[Thm.~1]{CoMaRo01}. Convex graphs with $r=\max_{i\in [n]}\deg(i)$ have treewidth at most $2r-1$. Unfortunately, the general algorithm of Courcelle, Makowsky and Rotics~\cite{CoMaRo01} is superexponential in the treewidth. An algorithm of F\"urer~\cite[Thm.~3]{Fu14}, for counting independent sets in graphs of bounded treewidth, could also be applied, since the treewidth of the line graph of a convex graph can be bounded by $8r^2$. (We will not prove these facts about treewidth here, since we do not use them.) Combined with F\"urer's algorithm, this gives an algorithm for the permanent which is linear in $n$, but exponential~in~$r^2$.

However, we will not use either of these approaches, since the following dynamic programming algorithm has better time complexity for the problem at hand.
\begin{lemma}\label{med:lem005}
Let $G=([n] \cup [n]',E)$ be a convex graph  containing a  perfect matching, and let $r=\max_{i\in [n]}\deg(i)$. Then, for any subgraph $G^*$ of $G$, the permanent of $A(G^*)$ can be evaluated exactly in time $O(r^{2r}n)$. Hence the permanent can be evaluated in polynomial time for all convex graphs with degree bound  $O(\log n/\log\log n)$.
\end{lemma}
\begin{proof}
Let $A=A(G^*)$. The algorithm uses triangular windows $W_i$ of width $2r+1$ and height $2r+1$, with corners at $A(i,(i-r)')$, $A(i,(i+r)')$ and $A(i+2r,(i+r)')$. Note, from Lemma~\ref{med:lem004},  that $W_i$ cuts $G$ as shown below.
(When $i\leq r$ the window overhangs the left boundary of the matrix and when $i>n-2r$ it overhangs
the bottom boundary.)
Moreover, for every edge of $G$ there is an index $i$ such that the corresponding entry of $A$ appears in the window $W_i$.
\begin{figure}[h]
\centering{%
\begin{tikzpicture}
\fill[white!92.5!black] (-0.1,4.5)--(3,4.5)--(4.5,3)--(4.5,-0.1)--cycle;
\draw[fill=white!80!black,densely dotted] (0,4)--(4,4)--(4,0)--cycle ;
\draw (0,4.5)--(4.5,0) (3,4.5)--(4.5,3);
\draw[line width=0.5pt]
 (4.7,0.1)edge[<-](4.7,1.3)  (4.7,1.5) node[empty] {\scriptsize$2r-1$} (4.7,1.7)edge[->](4.7,2.9);
  \draw (4,4.7) node[empty] {\scriptsize$(i+r)'$}
  (2,4.74) node[empty] {\scriptsize$i'$}  (0,4.7) node[empty] {\scriptsize$(i-r)'$}
  (-0.6,0) node[empty] {\scriptsize$i+2r$}   (-0.6,2) node[empty] {\scriptsize$i+r$}
  (-0.6,4) node[empty] {\scriptsize$i$}  (3,2.75) node[empty] {$W_i$};
  \draw[loosely dotted] (2,4.5)--(2,0) (0,2)--(4.5,2) (4,4.5)--(4,4) (4.5,4)--(4,4)
  (0,4.5)--(0,0) (0,0)--(4.5,0)  ;
\end{tikzpicture}
}
\end{figure}
At iteration $i$ of the algorithm, a \emph{subperfect} matching $Q$ will be a matching of $G^*$, such that
\begin{enumerate}
  \item Every row $j\leq i$ has a matching edge;
  \item Every column $j'\leq\min\{(i+r)',n'\}$ has a matching edge;
  \item No row $j>i+2r$ has a matching edge;
  \item No column $j'>(i+r)'$ has a matching edge.
\end{enumerate}

Note that a subperfect matching cannot always be extended to a perfect matching of $G^*$.
We consider the set
\[ S_i\ =\ \{M: M=Q\cap W_i\textrm{ and }Q\textrm{ is a subperfect matching}\,\}\,. \]
Note that $|S_i| < (2r)!$, since each column of $W_i$ is either empty or contains a unique edge in any of positions
$1,2,\ldots,j$, for $j=1,2,\ldots, 2r-1$.
For $M\in S_i$, let
\[ N_i(M)= \big|\{Q: Q\cap W_i=M\,\}\big|\,,\]
be the number of subperfect matchings represented by $M$.
Initially, $i=1$ and $S_1$ will be the set of all matchings in $W_1$ such that every vertex $j'\leq (r+1)'$ has a matching edge. When $i=n-r$, all the subperfect matchings represented in $W_{n-r}$ will be perfect matchings, and so we will have
\[  \per(A)\ =\ \sum_{M\in S_{n-r}} N_{n-r}(M)\,.\]
We must show how to update the $M$ and $N_i(M)$ from $W_i$ to $W_{i+1}$. Let $W^*_i=W_i\cap W_{i+1}$.
\begin{figure}\vspace{-3ex}
\centering{%
\begin{tikzpicture}[line width=1pt]
\fill[lightgray!60!white] (0.5,4)--(3.5,4)--(3.9,3.6)--(0.9,3.6)--cycle;
\fill[lightgray!60!white] (4,0.5)--(4,3.5)--(4.4,3.1)--(4.4,0.1)--cycle;
\begin{scope}
\draw[fill=none,densely dotted] (0,4)--(4,4)--(4,0)--cycle ;
\draw[line width=0.5pt] (0,4.5)--(4.5,0) (3,4.5)--(4.5,3);
\end{scope}
\begin{scope}[xshift=4mm,yshift=-4mm]
\draw[fill=none,densely dotted] (0,4)--(4,4)--(4,0)--cycle ;
\draw[line width=0.5pt] (0,4.5)--(4.5,0) (3,4.5)--(4.5,3);
\end{scope}
\draw (-1,3.8) node[empty] {\small row $i$}  (-0.5,3.8)edge[->,line width=0.5pt](0,3.8)
(3,2.5) node[empty] {\small$W^*_i$}
(4.2,4.1)edge[<-,line width=0.5pt](4.2,4.6) (4,4.8) node[empty] {\small column $(i+r+1)'$};
\end{tikzpicture}
}\vspace{1ex}
\end{figure}
First we remove row $i$. We remove all $M\in S_i$ such that row $i$ contains no matching edge, since they cannot
correspond to a subperfect matching at iteration $(i+1)$. Then we delete the matching edge in row $i$ from $M$,
for all $M\in S_i$. This will produce a set $S^*_i$ of matchings in $W^*_i$,
\[ S^*_i\ =\ \{M: M=Q\cap W^*_i\textrm{ and }Q\textrm{ is a subperfect matching}\,\}\,. \]
We must now add column $(i+r+1)'$ to $W_{i+1}$. For all $M^*\in S^*_i$,
we attempt to augment each $M^*$ with a matching edge $e$ in column $(i+r+1)'$.
Note that $e$ must be in $W_{i+1}$, and $e$ can be in any row which has no
matching edge in $M^*$. If no such row exists, we delete $M^*$ from $S^*_i$,
since it cannot correspond to a subperfect matching at iteration $i+1$.
Otherwise, for each possible choice of $e$, we add $M=M^*\cup\{e\}$ to
$S_{i+1}$, and set
\[ N_{i+1}(M)\ =\ \sum\,\{\thsp N_i(M^*) : M^* \in S_i,\ M^* \cap W_{i+1} = M \cap W_i\thsp\}. \]
This completes the description of the algorithm.

The operations in the update require $O(r|S_i|)$ time, except for the removal of duplicates, which can be implemented in $O(|S_i|\log |S_i|)= O(r^2|S_i|)$ time. Therefore, since
\[r^2|S_i|\,\leq\, r^2(2r)! \, \sim\, 2\sqrt{\pi}\,r^{5/2}(2r/e)^{2r}\, =\,O(r^{2r}), \]
and $O(n)$ updates must be performed, the overall time complexity of the algorithm is $O(r^{2r}n)$. This is polynomial in $n$ if $r=O(\log n/\log\log n)$.
\end{proof}
We can extend the algorithm of Lemma~\ref{med:lem005} to sample a matching uniformly at random.
To do this, we must retain the sets $S_i$ and the counts $N_i(M)$ ($M\in S_i$) used in the permanent evaluation.
Then the  sampling algorithm is a standard dynamic programming traceback through $S_{n-r}$,~\ldots,~$S_i$,~\ldots,~$S_1$, using the $N_i(M)$ to select matchings with the correct probability. See~\cite{Dyer03} for a more complete description of similar uses of traceback sampling. The time complexity for sampling a single matching is $O\big(\sum_i|S_i|\big)=O\big((2r)!\thsp n\big)$.

Thus  dynamic programming seems preferable to Markov chain methods for sampling perfect matchings from convex graphs with small degree bound, at least if the chain must be run for its a guaranteed mixing time.

\subsection{Biconvex graphs}\label{sec:biconvex}
\DGH~\cite{DiGrHo01} considered the following subclass of convex graphs.
A graph $G=([m]\cup[n]',E)$ is \emph{biconvex} if it is convex and $\nbh(j')$ is an interval $[\alpha_{j'},\beta_{j'}]\subseteq [n]$ for all $j'\in [n]'$. Thus $A(G)$ has the consecutive~1's property for both rows and columns.
\begin{lemma}\label{med:lem006}
Biconvex graphs are a proper hereditary subclass of convex graphs.
\end{lemma}
\begin{proof}
It is easy to see that the class \textsc{Biconvex} is a hereditary subclass of \textsc{Convex}. To see that it is a proper subclass, consider the example:
\[ \renewcommand{\arraystretch}{1.0}\kbordermatrix{ & 1' & 2' & 3' & 4'\\1 & 1 & 0 & 0 & 0\\2 & 1 & 1 & 1 & 0\\3 & 0 & 1 & 0 & 0\\4 & 0 & 0 & 1 & 1}\,. \]
In a biconvex ordering, row 2 must be adjacent to row 1 and row 3. Otherwise, columns $1'$ and $2'$ cannot be convex. But row 4 must also be adjacent to row 2, or column $3'$ cannot be convex. These conditions clearly cannot be satisfied simultaneously.
\end{proof}
As with convex graphs, it is possible to give excluded subgraph and excluded submatrix characterisations of biconvex graphs. Since these are a little easier to describe than for convex graphs, we will give the excluded subgraph characterisation. Tucker shows~\cite[Thm.~10]{Tucker72} that a bipartite graph is
biconvex if and only if it does not contain the graphs $\mathrm{I}_n$
for $n \ge 1$, $\mathrm{II}_1$, $\mathrm{II}_2$, $\mathrm{III}_1$,
$\mathrm{III}_2$ and $\mathrm{III}_3$ as induced subgraph. Here
$\mathrm{I}_n$ is a chordless cycle $C_{2n+4}$, $\mathrm{II}_1$ is the
\emph{triomino} and $\mathrm{III}_1$ is the \emph{tripod}.

{\tikzset{every node/.style=snode,every picture/.style={scale=0.3}}
\begin{figure}[h]
\begin{center}
  \begin{tikzpicture}
    \path (2,4) node (1b) {}
          (0,4) node (1w) {}
          (0,2) node (2b) {}
          (0,0) node (2w) {}
          (2,0) node (3b) {}
          (4,0) node (3w) {}
          (4,2) node (4b) {}
          (2,2) node (4w) {};
    \draw (1b) -- (1w) -- (2b) -- (2w) -- (3b) -- (3w) -- (4b) -- (4w) -- (1b);
    \draw (2b) -- (4w) -- (3b);
  \end{tikzpicture}
  \hspace*{7mm}
  \begin{tikzpicture}
    \path (0,1) node (c) {}
    (1.72,0.5) node (a1) {}
    (3.46,0) node (a2) {}
    (-1.72,0.5) node (b1) {}
    (-3.46,0) node (b2) {}
    (0,2.5) node (c1) {}
    (0,4) node (c2) {};
     \draw  (b2)--(b1)--(c)--(a1)--(a2) (c)--(c1)--(c2) ;
  \end{tikzpicture}
  \hspace*{7mm}
  \begin{tikzpicture}
    \path (0,1) node (1b) {}
          (0,3) node (1w) {}
          (2,3) node (2b) {}
          (3,4) node (2w) {}
          (4,3) node (3b) {}
          (6,3) node (3w) {}
          (6,1) node (4b) {}
          (4,1) node (4w) {}
          (3,0) node (5b) {}
          (2,1) node (5w) {};
    \draw (1b) -- (1w) -- (2b) -- (2w) -- (3b) -- (3w);
    \draw (3w) -- (4b) -- (4w) -- (5b) -- (5w) -- (1b);
    \draw (2b) -- (4w) -- (3b) -- (5w) -- (2b);
  \end{tikzpicture}
  \hspace*{7mm}
  \begin{tikzpicture}
    \path (6,4) node (1b) {}
          (4,4) node (1w) {}
          (2,4) node (2b) {}
          (0,2) node (2w) {}
          (2,0) node (3b) {}
          (4,0) node (3w) {}
          (6,0) node (4b) {}
          (6,2) node (5b) {}
          (4,2) node (5w) {};
    \draw (1b) -- (1w) -- (2b) -- (2w) -- (3b) -- (3w) -- (4b);
    \draw (2b) -- (5w) -- (3b)    (5w) -- (5b);
  \end{tikzpicture}
  \hspace*{7mm}
  \begin{tikzpicture}
    \path (6,4) node (1b) {}
          (4,4) node (1w) {}
          (2,4) node (2b) {}
          (0,4) node (2w) {}
          (0,2) node (3b) {}
          (0,0) node (3w) {}
          (2,0) node (4b) {}
          (4,0) node (4w) {}
          (6,0) node (5b) {}
          (2,2) node (6w) {}
          (4,2) node (6b) {};
    \draw (1b) -- (1w) -- (2b) -- (2w) -- (3b) -- (3w) -- (4b) -- (4w) -- (5b);
    \draw (2b) -- (6w) -- (4b)    (3b) -- (6w) -- (6b);
  \end{tikzpicture}
  \end{center}\vspace{5pt}
  \caption{The triomino, the tripod and the graphs $\mathrm{II}_2$, $\mathrm{III}_2$ and $\mathrm{III}_3$}
  \label{fig:Tucker}
\end{figure}
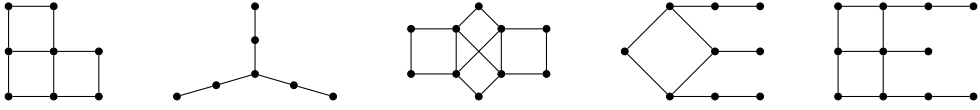}

We know from Lemma~\ref{med:lem001} that the switch chain converges eventually on biconvex graphs, but how quickly is this guaranteed to occur\thsp? Unfortunately, the convergence may be exponentially slow. Both Matthews~\cite{Matthe08} and Blumberg~\cite{Blumbe12} gave the following examples $\mathscr{G}_k=\big(\thsp[n]\cup[n]',\mathscr{E}_k\thsp\big)$, where $n=2k-1$\,:
\[ (i,j')\in\mathscr{E}_k\ \Longleftrightarrow\ \left\{
                \begin{array}{l}
                  1\leq i < k\mbox{  and } k'\leq j'\leq (k+i)'; \\[0.5ex]
                  i= k\mbox{  and } 1'\leq j'\leq n'; \\[0.5ex]
                  k < i \leq n \mbox{  and } (i-k)'\leq j'\leq k'.
                \end{array}
              \right.\]
For example, $\mathscr{G}_4$ has biadjacency matrix\vspace{-1ex}
\[ A(\mathscr{G}_4)\ =\ \kbordermatrix{%
  & 1' & 2'& 3' & 4' & 5'& 6'& 7' \\
1 & 0  & 0 & 0  & 1  & 1 & 0 & 0 \\
2 & 0  & 0 & 0  & 1  & 1 & 1 & 0 \\
3 & 0  & 0 & 0  & 1  & 1 & 1 & 1   \\
4 & 1  & 1 & 1  & 1  & 1 & 1 & 1   \\
5 & 1  & 1 & 1  & 1  & 0 & 0 & 0  \\
6 & 0  & 1 & 1  & 1  & 0 & 0 & 0  \\
7 & 0  & 0 & 1  & 1  & 0 & 0 & 0
} \vspace{1ex}\]
Let $\pi$ be any perfect matching. Then choosing $\pi'_k\leq k'$ forces $\pi'_i=(k+i)'$ for $i\in[k-1]$, and similarly choosing $\pi'_k\geq k'$ forces $\pi'_{k+i}=i'$ for $i\in[k-1]$.
Thus the set of perfect matchings of $\mathscr{G}_k$ is $S_1\cup S_2$,
where $S_1=\{\pi:\pi'_k\leq k'\}$ and $S_2=\{\pi:\pi'_k\geq k'\}$.

Clearly $S_1\cap S_2=\{\pi:\pi'_k=k'\}=\{\sigma\}$, for a single matching $\sigma$.
Moreover, it is not difficult to show that there are $2^{k-1}$ ways to extend
a partial matching $\pi$ with $\pi'_i=(k+i)'$ for $i\in[k-1]$ to a perfect
matching.  One way is to note that the submatrix induced by rows $[k,n]$ and columns
$[k]'$ is a \emph{chain graph}, for which the permanent is easy to compute.
(See Section~\ref{sec:chain} below and the formula presented there.)
Thus $|S_1\cap S_2|=1$ and $|S_1|=|S_2|=2^{k-1}$, and hence $|S_1\cup S_2|=2^k-1$.

Therefore, if the switch chain is started at a random matching in $S_1$, it will need $\Omega(2^n)$ time before it
reaches $\sigma$, and it cannot enter $S_2$ before this occurs. This gives an $\Omega(2^n)$ lower bound on the mixing time of the chain. This argument can be made completely rigorous, see~\cite{Blumbe12} or~\cite{Matthe08} for details.

\subsection{Monotone graphs}\label{sec:monotone}

\DGH~\cite{DiGrHo01} considered a structured subclass of biconvex graphs, which they called \emph{monotone}, and showed that the switch chain is ergodic on graphs in this class. However, note that Lemma~\ref{med:lem001} gives a stronger result, for the larger class of chordal bipartite graphs. \DGH~\cite{DiGrHo01} conjectured further that the switch chain mixes rapidly in the class \textsc{Monotone}.

A bipartite graph $G=([m]\cup[n]',E)$ will be called \emph{monotone} if it is isomorphic to a convex graph such that $\alpha'_i\leq\alpha'_j$ and $\beta'_i\leq\beta'_j$ for all $i,j\in[m]$ with $i<j$. Thus $A(G)$ has a ``staircase'' structure. An example is shown in Fig.~\ref{med:fig012}.
\begin{figure}[htb]
\centering{%
$G:\quad$\scalebox{0.9}{\raisebox{-1.45cm}{\begin{tikzpicture}[scale=1,font=\small,every node/.style=empty]
\draw (0,0) node[rnode] (1) {} +(-.3,0) node {1}
++(1,1)  node[cnode] (2') {} +(0,.3) node {$2'$}
++(1,0)  node[rnode] (3) {} +(0,.3) node {3}
++(1,0)  node[cnode] (4') {} +(0,.3) node {$4'$}
++(1,0)  node[rnode] (5) {} +(0,.3) node {5}
++(1,-1)  node[cnode] (5') {} +(.3,0) node {$5'$}
++(-1,-1)  node[rnode] (4) {} +(0,-.3) node {4}
++(-1,0)  node[cnode] (3') {} +(0,-.3) node {$3'$}
++(-1,0)  node[rnode] (2) {} +(0,-.3) node {2}
++(-1,0)  node[cnode] (1') {} +(0,-.3) node {$1'$} ;
\draw (1)--(1')  (1)--(2') (2)--(1') (2)--(2') (2)--(3') (3)--(1') (3)--(2')
(3)--(3') (3)--(4') (4)--(3') (4)--(4') (4)--(5') (4')--(5)--(5') ;
\end{tikzpicture}}}\hspace{2cm}
$ A(G):\ \ \kbordermatrix{%
    & 1' & 2' & 3' & 4'&5'  \\
1   & 1 & 1 & 0 & 0& 0 \\
2   & 1 & 1 & 1 & 0 & 0 \\
3  & 1 & 1 & 1 & 1 &  0 \\
4  & 0 & 0 & 1 & 1 & 1\\
5  &  0 & 0 & 0 & 1 & 1
} $}
\caption{A monotone graph}\label{med:fig012}
\end{figure}
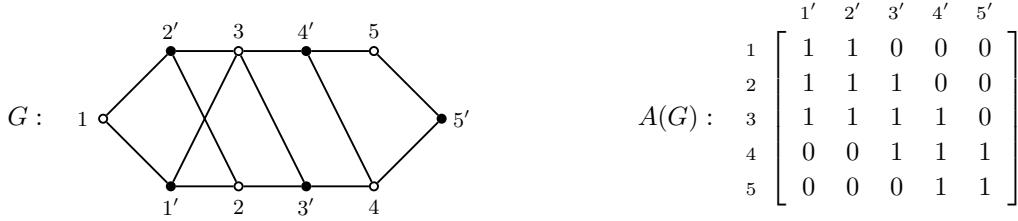

We will always assume that $G$ has the labelling of the definition, unless otherwise stated. First we show that, if $G$ is row-monotone, it is also column-monotone.
Recall that in a biconvex graph
$\nbh(j')$ is an interval $[\alpha_{j'},\beta_{j'}]\subseteq [n]$ for all $j'\in [n]'$.
\begin{lemma}\label{med:lem007} A  monotone graph is biconvex, and we have $\alpha_{i'}\leq\alpha_{j'}$, $\beta_{i'}\leq\beta_{j'}$ when $i',\,j'\in [n]'$ and $i'<j'$.
\end{lemma}
\begin{proof}
For $j\in [n]$, let $s=\min\{i\in\nbh(j')\}$ and  $t=\max\{i\in\nbh(j')\}$. If $s < i <t$, then $j'\geq\alpha'_{t}\geq\alpha'_i$ and  $j'\leq\beta'_{s}\leq\beta'_i$, so $j'\in[\alpha'_i,\beta'_i]=\nbh(i)$ and hence $i\in\nbh(j')$. Thus $\nbh(j')$ is the interval $[s,t]$, so we may take $\alpha_{j'}=s$, $\beta_{j'}=t$.
Hence $\alpha_{i'}=\min\{k:i'\in[\alpha_k',\beta_k']\}$ and $\alpha_{j'}=\min\{k:j'\in[\alpha_k',\beta_k']\}$, so $i'<j'$ implies $\alpha_{i'}\leq\alpha_{j'}$. Similarly, noting $\beta_{i'}=\max\{k:i'\in[\alpha_k',\beta_k']\}$ and $\beta_{j'}=\max\{k:j'\in[\alpha_k',\beta_k']\}$, we see that $i'<j'$ implies $\beta_{i'}\leq\beta_{j'}$.
\end{proof}

Next we show a ``forbidden submatrix'' characterisation of monotone graphs, extending that of Lubiw~\cite{Lubiw87} for chordal bipartite graphs.
 \begin{lemma}\label{med:lem008}
A bipartite graph is monotone if and only if it is isomorphic to a graph $G$ such that $A(G)$
has none of the following as an induced $2\times 2$ submatrix:
   \begin{center}
   \forba\ \textup{(Gamma)\;:}\hspace*{5pt}$\begin{bmatrix} 1 & 1\\ 1 & 0 \end{bmatrix}$,\hspace*{15pt}
   \forbb\ \textup{(backwards L)\;:}\hspace*{5pt}$\begin{bmatrix} 0 & 1\\ 1 & 1 \end{bmatrix}$,\hspace*{15pt}
   \forbc\ \textup{(slash)\;:}\hspace*{5pt}$\begin{bmatrix} 0 & 1\\ 1 & 0 \end{bmatrix}$.

   \end{center}
 \end{lemma}
\begin{proof}
Suppose $G$ is monotone, but $A$ contains \forba or \forbc in rows $i$ and $j$, with $i<j$. Then row-convexity implies $\beta'_i>\beta'_j$, a contradiction. Similarly, if $A(G)$ contains a $\,\forbb$\thsp, then row-convexity implies $\alpha'_i>\alpha'_j$, again a contradiction. Thus, if $G$ is a monotone graph, $A(G)$ cannot contain \forba, \forbb or \forbc.

Now assume $A(G)$ contains no \forba, \forbb\ or \forbc. Suppose $\nbh(i)$ is not an interval, so there exist $j'<k'<l'$ so that $(i,j'),\,(i,l')\in E$, but $(i,k')\notin E$. Since $\nbh(k')\neq\emptyset$, there exists $s\in [n]$ such that $(s,k')\in E$. If $s<i$, then $A(G)$ contains the first configuration below, which is either a \forba\ or a \forbc, a contradiction. If $s>i$, then $A(G)$ contains the second configuration below, which is either a \forbb\ or a \forbc, also a contradiction.\vspace{5pt}

\centerline{%
$\kbordermatrix{
      & j'  & k' \\
    s & ? & 1 \\
    i & 1 & 0}$\hspace{2cm}
   $\kbordermatrix{
        & k'& l' \\
        i & 0 & 1\\
    s  & 1 & ?}$
}\vspace{5pt}

Therefore suppose that $i<j$, but $\alpha'_i>\alpha'_j$. Then $A(G)$ contains the first configuration below, which is a \forbb\ or \forbc, a contradiction. Similarly, if $\beta'_i>\beta'_j$, $A(G)$ contains the second configuration below, which is a \forba\ or \forbc, again a contradiction. Hence $G$ is monotone.\vspace{5pt}

\centerline{%
\hspace*{2.9cm}$\kbordermatrix{
        & \alpha'_j  & \alpha'_i \\
    i & 0 & 1 \\
    j & 1 & ?}$\hspace{2cm}
   $\kbordermatrix{
        & \beta'_j & \beta'_i \\
        i & ? & 1\\
    j  & 1 & 0}$\hspace*{2cm}\raisebox{-10pt}{\qed}
}
\end{proof}
A \emph{bipartite permutation graph} is a permutation graph which is also bipartite. A graph $G=(V,E)$ is a \emph{permutation graph} if there are permutations $\pi,\sigma$ of $V$ so that $(\pi_i,\pi_j)\in E$ if and only if $\pi_i<\pi_j$ and $\sigma_i>\sigma_j$. This can be given  an \emph{intersection} presentation, where $\pi,\sigma$ are on parallel lines, and connected by lines $(v,v)$, for  all $v\in V$. Then $(v,w)\in E$  if and only if corresponding lines $(v,v)$ and $(w,w)$ cross.  An example is shown in Fig.~\ref{med:fig003} below.

Spinrad, Brandst\"{a}dt and Stewart~\cite{SpBrSt87} studied this class of graphs,
and gave $O(\thsp|E|\thsp)$ time algorithms for recognising membership in the class, and for constructing a crossing representation.
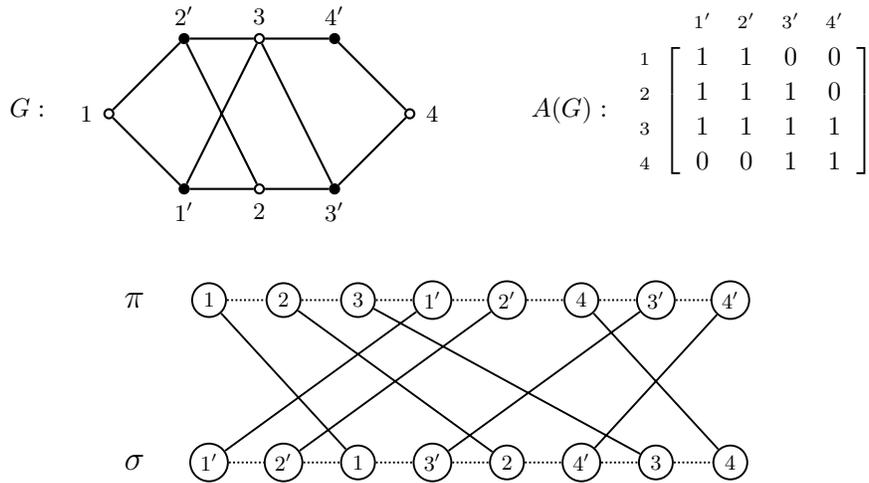
\begin{figure}[th]
\centering{%
$G:\quad$\raisebox{-1.5cm}{\begin{tikzpicture}[scale=1,font=\small,every node/.style=empty]
\draw (0,0) node[rnode] (1) {} +(-.3,0) node {1}
++(1,1)  node[cnode] (2') {} +(0,.3) node {$2'$}
++(1,0)  node[rnode] (3) {} +(0,.3) node {3}
++(1,0)  node[cnode] (4') {} +(0,.3) node {$4'$}
++(1,-1)  node[rnode] (4) {} +(.3,0) node {4}
++(-1,-1)  node[cnode] (3') {} +(0,-.3) node {$3'$}
++(-1,-0)  node[rnode] (2) {} +(0,-.3) node {2}
++(-1,0)  node[cnode] (1') {} +(0,-.3) node {$1'$} ;
\draw (1)--(1')  (1)--(2') (2)--(1') (2)--(2') (2)--(3') (3)--(1') (3)--(2')
(3)--(3') (3)--(4') (4)--(3') (4)--(4') ;
\end{tikzpicture}}\hspace{1cm}
$ A(G):\ \ \kbordermatrix{%
    & 1' & 2' & 3' & 4'  \\
1   & 1 & 1 & 0 & 0 \\
2   & 1 & 1 & 1 & 0  \\
3  & 1 & 1 & 1 & 1 \\
4  &  0 & 0 & 1 & 1
} $}\\[3ex]
\centering{\ \\[1ex]%
\scalebox{0.9}{\begin{tikzpicture}[xscale=0.55,yscale=0.6,font=\small]
\draw (0,0) node[empty] {\Large$\pi$} ++(2,0) node[circlenode] (1)  {1} ++(2,0)  node[circlenode] (2) {2}
++(2,0)  node[circlenode] (3) {3} ++(2,0)  node[circlenode] (1') {$1'$}
++(2,0)  node[circlenode] (2') {$2'$} ++(2,0)  node[circlenode] (4) {4} ++(2,0)
node[circlenode] (3') {$3'$} ++(2,0)  node[circlenode] (4') {$4'$}
(0,-4) node[empty] {\Large$\sigma$} ++(2,0) node[circlenode] (p1')  {$1'$} ++(2,0)  node[circlenode] (p2') {$2'$}
++(2,0)node[circlenode] (p1) {1} ++(2,0)  node[circlenode] (p3') {$3'$}
++(2,0)  node[circlenode] (p2) {2} ++(2,0)  node[circlenode] (p4') {$4'$}
++(2,0)  node[circlenode] (p3) {3} ++(2,0)  node[circlenode] (p4) {4};
\draw (1)--(p1) (2)--(p2) (3)--(p3) (4)--(p4) (1')--(p1') (2')--(p2') (3')--(p3') (4')--(p4') ;
\draw[densely dotted] (1)--(2)--(3)--(1')--(2')--(4)--(3')--(4') (p1')--(p2')--(p1)--(p3')--(p2)--(p4')--(p3)--(p4);
\end{tikzpicture}}
}\vspace{5mm}
\caption{A bipartite permutation graph with its intersection representation}\label{med:fig003}
\end{figure}

Our reason for introducing this class of graphs is that the bipartite permutation graphs are precisely the monotone graphs.
\begin{lemma}\label{med:lem009}
A graph is monotone if and only if it is a bipartite permutation graph.
\end{lemma}
\begin{proof}
The condition of Lemma~\ref{med:lem008} is equivalent to the following. If $(i,k'), (j,\ell')\in E$ are such that $i<j$ and $k'>\ell'$, then $(i,\ell')$, $(j,k')\in E$. The conclusion now follows from the characterisation of bipartite permutation graphs given in~\cite{SpBrSt87}, in particular Definition~3 and Theorem~1.
\end{proof}
Note that Lemma~\ref{med:lem008} is not a ``forbidden subgraph'' characterisation in the usual graph-theoretic sense. However, such a characterisation is known.
\begin{lemma}\textup{\cite[Lem.\,1.46]{Kohl99}}\label{med:lem010}
A graph is monotone if and only if it is chordal bipartite (i.e.~it has no chordless cycle of length other than 4), and it contains none of the three graphs shown in Fig.~\ref{med:fig004} as an induced subgraph.\qed
\end{lemma}\vspace{-2ex}
\begin{figure}[h]
\tikzset{every node/.style={circle,draw,inner sep=0pt,minimum size=1.25mm}}
\begin{center}
  \begin{tikzpicture}[scale=0.67]
    \begin{scope}
    \path (0:0) node[fill=black] (c) {} (-30:1) node (a1) {} (-30:2) node[fill=black] (a2) {}
    (90:1) node (b1) {} (90:2) node[fill=black] (b2) {}  (210:1) node (c1) {} (210:2) node[fill=black] (c2) {};
     \draw (a2)--(a1) (b1)--(c)--(a1) (c)--(c1) (b2)--(b1) (c2)--(c1) ;
  \end{scope}
\begin{scope}[xshift=5cm,yshift=-1cm]
    \path  (2,0) node[fill=black] (a) {} (0,0) node (b) {} (2,0.75) node (c) {} (0,0.75) node[fill=black] (d) {}
     (2,1.5) node[fill=black] (e) {} (0,1.5) node (f) {}  (0,3) node[fill=black] (g) {} ;
     \draw  (g)--(f)--(e)--(c)--(d)--(b) (a)--(c) (d)--(f) ;
  \end{scope}
\begin{scope}[xshift=10cm,yshift=-1cm,xscale=0.75]
    \path (0,0) node (a) {} (1,0) node[fill=black] (b) {} (2,0) node (c) {}
    (0,1) node[fill=black] (d) {} (1,1) node (e) {}  (2,1) node[fill=black] (f) {} (1,3) node[fill=black] (g) {};
     \draw (a)--(b)--(c)--(f)--(e)--(d)--(a) (b)--(e)--(g) ;
  \end{scope}
  \end{tikzpicture}\\[2ex]
  \caption{The tripod, the armchair and the stirrer.}
  \label{med:fig004}
\end{center}\vspace{-2ex}
\end{figure}
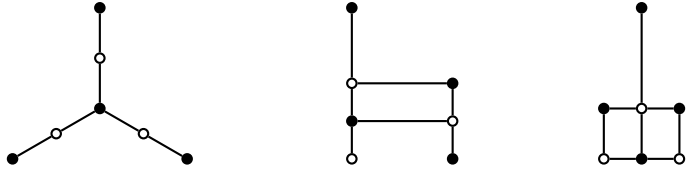

For example, the graph $G$ given in Fig.~\ref{med:fig005} contains the armchair as a subgraph.
\begin{figure}[ht]
\centering{%
{\renewcommand{\arraystretch}{0.8}
$A(G)\,=\,$ $\kbordermatrix{%
    & 1' & 2' & 3' & 4' & 5' \\
1 & 0 & 0 & 0 & 1 & 0  \\[0.5ex]
2  & 0 & 0 & 1 & 1 & 1  \\[0.5ex]
3  & 0 & 1 & 1 & 1 & 0  \\[0.5ex]
4  & 1 & 1 & 1 & 0 & 0  \\[0.5ex]
5  & 0 & 1 & 1 & 0 & 0
}$
\hspace{1cm}\begin{tabular}{c}
induced\ \\subgraph\,:
\end{tabular}
$\kbordermatrix{%
& 3' & 4' & 5' \\
1  & 0 & 1 & 0  \\[0.5ex]
2  & 1 & 1 & 1  \\[0.5ex]
3  & 1 & 1 & 0  \\[0.5ex]
4  & 1 & 0 & 0
}$ \hspace{10mm}
\tikzset{every node/.style={empty}}
\raisebox{-12mm}{\begin{tikzpicture}[line width=1pt,font=\small,yscale=1.1,xscale=1.5,label distance=-2pt]
    \path (0,0.5) node[cnode,label=left:$2\,$] (2) {}
    (0,0) node[rnode,label=left:$5'$] (5') {}
    (1,1) node[cnode,label=right:$3$] (3)  {}
    (1,0.5) node[rnode,label=right:$4'$] (4') {}
    (1,0) node[cnode,label=right:$1$] (1) {}
    (0,1) node[rnode,label=left:$3'$] (3') {}
    (0,2) node[cnode,label=left:$4$] (4)  {} ;
     \draw  (1)--(4') (4')--(2) (2)--(5') (4')--(3)--(3') (2)--(3') (3')--(4);
  \end{tikzpicture}}}}
  \caption{A biconvex graph containing the armchair}
  \label{med:fig005}
\end{figure}

\begin{lemma}\label{med:lem011}
Monotone graphs are a proper hereditary subclass of biconvex graphs.
\end{lemma}
\begin{proof}
The hereditary property follows easily from the definitions.
The inclusion follows from Lemma~\ref{med:lem007},
and strict inclusion follows from the example of Fig.~\ref{med:fig005}.
\end{proof}

To apply the switch chain to a monotone graph, we need to know whether it contains any perfect matching. If it does, we need to identify one efficiently, in order to start the chain. However, these are easy questions.
\begin{lemma}\label{med:lem012}
A monotone graph $G=([n]\cup[n]',E)$ contains a perfect matching if and only if it contains
the \emph{diagonal matching} $\delta=\{(i,i'): i\in[n]\}$.
 \end{lemma}
\begin{proof}
We prove this by induction on $n$. If $n=1$, then $E=\{(1,1')\}$, and there is nothing to prove. So, suppose $n>1$.  Clearly $(1,1')\in E$, or else either $1$ or $1'$ is an isolated vertex, and hence $G$ has no perfect matching. We will show that there is a perfect matching $M^*$ which contains $(1,1')$.  Therefore, suppose that $M$ is any perfect matching in $G$, with $(1,1')\notin M$. Then $(1,j'),(i,1')\in M$ for some $i\geq 2$, $j'\geq 2'$, and we have $(1,1')\in E$. Hence $(i,j')\in E$, or else $A(G)$ would contain a \forba.\vspace{-1ex}
\[\kbordermatrix{
        & 1'  & j' \\
    1 & 1 & \one \\
    i &  \one & ?}\vspace{1ex}
   \]
Thus $M^*= M\setminus\{(1,j'), (i,1')\}\cup\{(1,1'), (i,j')\}$ is a perfect matching containing the edge $(1,1')$.
Now we use induction on the graph $G^*$ given by deleting $1$ and $1'$ from $G$, which contains the perfect matching $M^*\setminus\{(1,1')\}$.
\end{proof}

We will be particularly interested in \emph{Hamiltonian} monotone graphs, and there is also an easy criterion for Hamiltonicity of a monotone graph. To state this, we consider the graph illustrated below, the \emph{ladder} $L_n$.
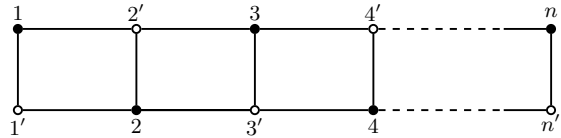
\begin{figure}[h]
\centering{%
\tikzset{every node/.style={empty}}
\scalebox{0.9}{\begin{tikzpicture}[xscale=1.75,yscale=1.2,font=\small]
\path (0,0) node[rnode] (1') {}  +(0,-0.2) node {$1'$}
++(1,0) node[cnode] (2) {}  +(0,-0.2) node {$2$}
++(1,0) node[rnode] (3') {}  +(0,-0.2) node {$3'$}
++(1,0) node[cnode] (4) {}  +(0,-0.2) node {$4$}
++(1.1,0) coordinate (u')
++(0.4,0) node[rnode] (n') {}  +(0,-0.15) node {$n'$}
++(0,1) node[cnode] (n) {}  +(0,0.2) node {$n$}
++(-0.4,0) coordinate (u) {}
++(-1.1,0) node[rnode] (4') {}  +(0,0.2) node {$4'$}
++(-1,0) node[cnode] (3) {}  +(0,0.2) node {$3$}
++(-1,0) node[rnode] (2') {}  +(0,0.2) node {$2'$}
++(-1,0) node[cnode] (1)  {}  +(0,0.2) node {$1$} ;
\draw (1)--(2')--(3) (1')--(2)--(3')  (n)--(n') (u)--(n) (u')--(n')
(1)--(1') (2')--(2) (2)--(3') (3)--(3') (4)--(4')  (3)--(4') (3')--(4);
\draw[dashed]  (4')--(u) (4)--(u');
\end{tikzpicture}}}\vspace{-2ex}
\end{figure}
\begin{lemma}\label{med:lem013}
 $L_n$ is a Hamiltonian monotone graph.
 \end{lemma}
\begin{proof}
Clearly $L_n$ is bipartite, and  $\nbh(1)$, $\nbh(2)$, $\nbh(3)$, \ldots,  $\nbh(n-1)$, $\nbh(n)$ are, respectively,
\begin{equation*}
   \{1',2'\},\,  \{1',2',3'\},\, \{2',3',4'\},\ \ldots\ ,\{(n-2)',(n-1)',n'\},\, \{(n-1)',n'\},
\end{equation*}
so are non-empty intervals satisfying the required ordering conditions. Finally, $L_n$ has the Hamilton cycle
$ 1'\to 2\to 3'\to\cdots\to n'\to n\to \cdots\to 3\to 2'\to 1\to1'$.
\end{proof}
Then we have the following.
\begin{lemma}\label{med:lem014}
A monotone graph $G$ is Hamiltonian if and only if it contains the ladder as a spanning subgraph.
 \end{lemma}
\begin{proof}
If $G$ has a spanning ladder, the Hamilton cycle in the ladder is also a Hamilton cycle in $G$, and so $G$ is Hamiltonian.

If $G=([m]\cup[n]',E)$ is Hamiltonian, it has a perfect matching,  so $m=n$ and $G$ contains the diagonal matching $\delta$, from Lemma~\ref{med:lem012}. We will show by induction that $G$ contains $L_n$, and so has a  spanning ladder. The base case is $n=2$. Then $G$ must be a 4-cycle, so $G=L_2$.

If $n>2$, consider any Hamilton cycle $H$ in $G$. Vertices $1$ and $1'$ lie on this cycle. There are two cases:
\begin{enumerate}
  \item The cycle $H$ contains the edge $(1,1')$. Let $j'\neq 1' $ be adjacent on $H$ to $1$,
and $i\neq 1$ be adjacent on $H$ to $1'$. Since $i,j\geq 2$,  biconvexity implies $(1,2')\in E$ and $(2,1')\in E$. Thus the three edges $(1,1')$, $(1,2')$, $(2,1')$ of $L_n$ are in $E$. Also $(i,j')\in E$, since $G$ is \gammafree. Hence $i\to j'\to\cdots\to i$ is a Hamilton cycle $H^*$ in the monotone graph $G^*$ obtained by deleting $1$ and $1'$ from $G$.

\begin{figure}[h]
\centerline{%
\raisebox{7mm}{$\kbordermatrix{
        & 1'  & j' \\
    1 & \squ & \one \\
    i &  \one & 1}$}\hspace*{3cm}
\begin{tikzpicture}[scale=1,every node/.style={empty}]
\path
(0,0) node[rnode] (1) {}  +(-0.25,0) node {$1$}
(0,1) node[cnode] (j') {} +(-0.25,0) node {$j'$} ;
\path
(2,0) node[cnode] (1') {}  +(0.25,0) node {$1'$}
(2,1) node[rnode] (i) {}  +(0.25,0) node {$i$} ;
\draw (j')--(1)--(1')--(i) ;
\draw[densely dashed] (j')edge[bend left=60](i) ;
\draw[densely dotted] (j')--(i) ;
\end{tikzpicture}}
\end{figure}
  \item The cycle $H$ does not contain the edge $(1,1')\in E$.  Let $j',l'$ be the vertices of $H$ adjacent to $1$, and $i,k$ the vertices of $H$ adjacent to $1'$, so that $H$ contains paths $i \to\cdots\to j'$ and $k\to\cdots\to l'$, avoiding $1$ and $1'$. Now, since $G$ is \gammafree, $(i,j'),\, (i,l'),\, (k,j'),\, (k,l')\in E$. Since $(1,1')\in E$, and $(1,j')\in E$ for some $j\geq 2$,  convexity implies that $(1,2')\in E$. Similarly, since $(1,1'),(i,1')\in E$, with $i\geq 2$, convexity implies that $(2,1')\in E$. Thus the three edges $(1,1')$, $(1,2')$, $(2,1')$ of $L_n$ are in $E$. Also $i\to\cdots\to j'\to k\to\cdots\to l'\to i$ is a Hamilton cycle $H^*$ in the monotone graph $G^*$ given by deleting $1$ and $1'$ from $G$.\vspace{-1ex}
\begin{figure}[h]
\centering{%
\raisebox{12mm}{$\kbordermatrix{
        & 1'  & j' & l' \\
    1 & 1 & \one & \squ\\
    i &  \one & 1  & 1\\
    k & \squ  & 1 & 1
   }$}\hspace*{3cm}
   \begin{tikzpicture}[scale=0.85,font=\small,every node/.style=empty]
\path
(0,0) node[rnode] (1) {}  +(-0.25,0) node {$1$}
(0.5,1) node[cnode] (j') {}  +(-0.25,0) node {$j'$}
(0.5,-1) node[cnode] (l') {}  +(-0.25,0) node {$l'$}
(3.5,0) node[cnode] (1') {}  +(0.25,0) node  {$1'$}
(3,1) node[rnode] (i) {}  +(0.25,0)  node {$i$}
(3,-1) node[rnode] (k) {}  +(0.25,0)  node {$k$} ;
\draw (j')--(1)--(l') (i)--(1')--(k) ;
\draw[densely dashed] (j')edge[bend left=45](i) (l')edge[bend right=45](k) ;
\draw[densely dotted] (1)--(1') (i)--(l') (k)--(j') (i)--(j') (k)--(l') ;
\end{tikzpicture}}
\end{figure}
\end{enumerate}\vspace{-2ex}
In both cases, the edges $(1,1')$, $(1,2')$, $(2,1')$ of $L_n$ are in $E$, and we have a Hamiltonian monotone graph $G^*$ with bipartition $[2,n]\cup[2,n]'$. It now follows by induction that $G$ contains $L_n$.
\end{proof}
\subsection{Chain graphs}\label{sec:chain}
\DGH called the simplest class of graphs they considered ``one-sided restriction'' graphs. These are usually called \emph{chain graphs} in the graph theory literature~\cite{Yannak81}, and form a proper subclass of monotone graphs.

A chain graph is isomorphic to a graph $G=([m]\cup[n]',E)$ where $\nbh(i)=[a_i]'$ for all $i\in[m]$, with $a_1 \leq a_2\leq \cdots\leq a_m$. Hence chain graphs are a subclass of monotone graphs, given by taking $\alpha'_i=1'$, $\beta'_i=a'_i$,  for all $i\in[n]$.
 An example is shown in Fig.~\ref{med:fig015}.
 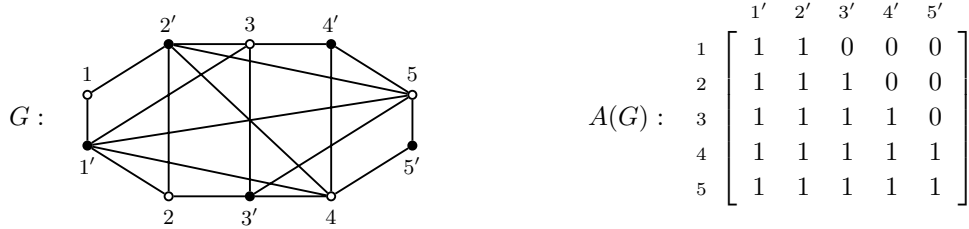
\begin{figure}[htb]
\centering{%
$G:\quad$\scalebox{0.9}{\raisebox{-1.6cm}{\begin{tikzpicture}[xscale=1.2,yscale=1.5,font=\small,
every node/.style=empty]
\tikzset{line width=0.75pt}
\draw (0,0.5) node[rnode] (1) {} +(0,.2) node {1}
++(1,0.5)  node[cnode] (2') {} +(0,.2) node {$2'$}
++(1,0)  node[rnode] (3) {} +(0,.2) node {3}
++(1,0)  node[cnode] (4') {} +(0,.2) node {$4'$}
++(1,-0.5)  node[rnode] (5) {} +(0,.2) node {5}
++(0,-0.5)  node[cnode] (5') {} +(0,-.2) node {$5'$}
++(-1,-0.5)  node[rnode] (4) {} +(0,-.2) node {4}
++(-1,0)  node[cnode] (3') {} +(0,-.2) node {$3'$}
++(-1,0)  node[rnode] (2) {} +(0,-.2) node {2}
++(-1,0.5)  node[cnode] (1') {} +(0,-.2) node {$1'$} ;
\draw (1)--(1')  (1)--(2') (2)--(1') (2)--(2') (2)--(3') (3)--(1') (3)--(2')
(3)--(3') (3)--(4') (4)--(1') (4)--(2') (4)--(3') (4)--(4') (4)--(5') (1')--(5) (5)--(2') (5)--(3') (4')--(5)--(5') ;
\end{tikzpicture}}}\hspace{2cm}
$ A(G):\ \ \kbordermatrix{%
    & 1' & 2' & 3' & 4'&5'  \\
1   & 1 & 1 & 0 & 0& 0 \\
2   & 1 & 1 & 1 & 0 & 0 \\
3  & 1 & 1 & 1 & 1 &  0 \\
4  & 1 & 1 & 1 & 1 & 1\\
5  &  1 & 1 & 1 & 1 & 1
} $}
\caption{A chain graph}\label{med:fig015}
\end{figure}

The following easy fact is then true.
\begin{lemma}\label{med:lem002}
  $\nbh(j')=[b_{j},m]$ for all $j'\in[n]'$ with $b_1 \geq b_2\geq \cdots\geq b_n$.
\end{lemma}
\begin{proof}
  Since $a_i\leq a_{i+1}$, $(i,j')\in E$ implies $(i+1,j')\in E$. Let $b_j=\min\{i:a_i\geq j\}$.
\end{proof}
We note here that chain graphs have a very simple excluded subgraph characterisation. A bipartite graph is a chain graph if and only if it does not contain $2K_2$, the graph with four vertices and two disjoint edges, as an induced subgraph.

\begin{lemma}\label{lem:chainvsmonotone}
Chain graphs are a proper hereditary subclass of monotone graphs.
\end{lemma}

\begin{proof}
The hereditary property and the inclusion are clear from the definition.
The inclusion is strict because a chain graph is a monotone graph with $\alpha'_i=1'$ for all $i\in[n]$. Hence the biadjacency matrix contains a column of 1's. But there exist monotone graphs which do not have a column of 1's. The graph in Fig.~\ref{med:fig012} is an example.
\end{proof}

\DGH~\cite{DiGrHo01} observed that there is a ``classical'' explicit formula for the permanent of a chain graph $G$. Of course, we must have $m=n$. Then, if $A=A(G)$,
\[ \per(A)\ =\ \left\{
                     \begin{array}{ll}
                       0, & \hbox{if $a_i<i$ for any $i\in[n]$;} \\
                       \prod_{i=1}^n (a_i-i+1), & \hbox{otherwise.}
                     \end{array}
                   \right.
\]
For example, if\vspace{-1ex}
\[ A:\ \ \kbordermatrix{%
       & 1' & 2' & 3' & 4'&5'  \\
1   & 1 & 1 & 0 & 0& 0 \\
2   & 1 & 1 & 1 & 0 & 0 \\
3  & 1 & 1 & 1 & 1 &  0 \\
4  & 1 & 1 & 1 & 1 & 1\\
5  &  1 & 1 & 1 & 1 & 1
} \,,\quad\mathrm{then}\ \ \per(A)\, =\,2(3-1)(4-2)(5-3)(5-4)\ =\ 16.\]
After noting that the first $a_1$ columns of $A$ are all 1's, and hence identical, this formula can be proved by an easy induction on the row order. The proof method can also be used to sample a perfect matching uniformly at random in $O(n)$ time.

Matthews~\cite{Matthe08} showed, using a coupling argument, that the mixing time of the switch chain for chain graphs is bounded by $O(n^3\log n)$. Blumberg~\cite{Blumbe12} gave a detailed study of the eigenvalues of the transition matrix of the switch chain for this class, based on earlier work of Hanlon~\cite{Hanlon96}.

These results clearly have little computational application, but establishing the mixing time of the switch chain for graphs in the the class \textsc{Chain} is far from trivial. For example, there are chain graphs for which the original Jerrum-Sinclair Markov chain~\cite{JerSin89} has exponential mixing time. Consider the graph $G$ for which $A(G)$ is lower triangular, so $A(i,j)=1$ if $i\leq j$, $A(i,j)=0$ otherwise. Then $\per(A)=1$, from the formula above, but the graph $G^*$ given by deleting vertices $1$ and $n'$ has $\per(A^*)=2^{n-3}$ by the same formula, where $A^*=A(G^*)$. Thus $G$ has one perfect matching, but an exponential number of near-perfect matchings. Therefore the algorithm of~\cite{JerSin89} will need exponential time to sample a perfect matching almost uniformly.

\section{Analysis of the switch chain}\label{sec:canonical}
In Section~\ref{sec:classes} we have seen that the hereditary graph classes considered by \DGH~\cite{DiGrHo01} form an ascending sequence:
\[ \textsc{Chain}\ \subset\ \textsc{Monotone}\ \subset\ \textsc{Biconvex}\ \subset\ \textsc{Convex}\ \subset\ \textsc{Chordal bipartite}.\]
We know from Lemma~\ref{med:lem001} that the switch chain is ergodic for (balanced) bipartite graphs in all these classes, and has diameter at most $n$. Here we consider the mixing time of the switch chain for these classes, that is, the time required to reach near-stationarity. We will define this more formally below.

Here we take the Computer Science viewpoint of trying to distinguish only between polynomial time (rapid) mixing, and superpolynomial time (slow) mixing. We have observed that the switch chain may have exponential mixing time in the class \textsc{Biconvex}, and the switch chain has mixing time $O(n^3\log n)$ in the class \textsc{Chain}~\cite{Matthe08}. Therefore we need only determine whether or not the class \textsc{Monotone} has polynomial mixing time. The remainder of this section will be devoted to showing that this class does indeed have polynomial mixing time.

It is usual to measure deviation from the stationary distribution
in terms of variation distance.
For an ergodic Markov chain on state space~$\Omega$, transition probabilities
$\tprob:\Omega^2\to[0,1]$ and stationary distribution~$\pi$,
the distance to stationary at time~$t$, starting in state $x\in\Omega$,
is $\Delta_x(t)=\max_{A\subset\Omega}|\tprob^t(x,A)-\pi(A)|$.
Then the {\it mixing time}, $\tmix(\epsilon)$, is the first time at which the Markov chain is
$\epsilon$-close to stationarity, maximised over the choice of starting state~$x$:
$$
\tmix(\epsilon)=\max_{x\in\Omega}\min\{t:\Delta_x(t)\leq\epsilon\}.
$$
Two remarks about this definition.  First, the function $\Delta_x(t)$ is monotonically
non-increasing in~$t$, so once the $t$-step distribution of the Markov chain
is $\epsilon$-close to stationarity, it remains that way.
Second, the dependence on of~$\tmix$ on~$\epsilon$ is weak, typically a
multiplicative factor $\log\epsilon^{-1}$, so it is usual to quote the mixing
time as a function of the size of the problem instance, in our case~$n$,
and ignore the dependence on $\epsilon$ (or set $\epsilon$ to some conventional
value, usually $\epsilon=1/e$).  For more on mixing time, and the general context
for this section of the paper, refer to \cite[Chaps.\,3 \&~4]{AldFil02} and \cite[Chaps.\,3 \&~5]{Jerrum03}.

\subsection{Canonical paths and flows}\label{sec:paths}
Although there are other approaches to bounding the mixing time of Markov chains, here
we will attempt only to apply the canonical paths approach of Jerrum and Sinclair~\cite{JerSin89}.
For any symmetric Markov chain, this may be described briefly as follows.

Suppose the problem size is $n$. The method requires constructing paths of transitions of the chain $X=Z_1\to Z_2\to\cdots\to Z_\ell=Y$, between every pair of states $X$ and $Y$ in the state space $\Omega$, such that the length $\ell$ of each path is at most polynomial in $n$.
This is often easy to achieve, but to obtain a good upper bound on mixing time it
is essential that the paths are ``spread out'' over the state space, and do not
overload any particular transition.  The degree of success in achieving this end is
measured by the congestion of the set of paths.
Denote the (canonical) path from $X$ to $Y$ by $\gamma_{XY}$.  Then the
congestion~$\varrho$ of the chosen paths is given by
\begin{equation}\label{eq:congestion}
\varrho=\max_{(Z,Z^\dag)}\bigg\{\frac1{\pi(Z)\tprob(Z,Z^\dag)}
\sum_{X,Y:\gamma_{XY}\ni(Z,Z^\dag)}\pi(X)\pi(Y)\,|\gamma_{XY}|\bigg\},
\end{equation}
where $|\gamma_{XY}|$ denotes the length of the path $\gamma_{XY}$,
and the maximisation is over all transitions $Z\to Z^\dag$, i.e., all pairs
with $\tprob(Z,Z^\dag)>0$ and $Z^\dag\not=Z$, and the sum is over all
paths that use the transition $Z\to Z^\dag$.

Polynomial mixing time will follow from a polynomial upper bound on congestion.
For the switch Markov chain, where $\pi$ is uniform over $\Omega$,
it can be seen from the definition of $\varrho$ that this is equivalent
to ensuring that the number of canonical paths through any transition
is bounded by $|\Omega|$, possibly multiplied by a polynomial factor.
A possible strategy for obtaining a good bound on congestion is therefore the following.
Fix a transition $Z\to Z^\dag$.  For every canonical path $\gamma_{XY}$ from $X$ to~$Y$ that
uses transition $(Z,Z^\dag)$, specify
an \emph{encoding} $W\in\Omega$, such that, given $W$ and $g$ additional bits of information,
we can identify $X$ and $Y$ uniquely. It is common to refer to the additional  information as ``guesses''.
In this setting, the mixing time $\tmix$ of the chain can be bounded by $2^{O(g)}\textrm{poly}(n)$.
Ideally, we seek $g=O(\log n)$, since this will give a polynomial bound on the mixing time.

The strategy just described is the one we will apply in this section to
the switch chain on monotone graphs.  In the proof of Lemma~\ref{med:lem001}
we implicitly constructed canonical paths for the more general situation
of chordal bipartite graphs.  The counterexample $\CG_k$ from Section~\ref{sec:biconvex}
demonstrates that the congestion of those canonical paths is in fact exponentially large,
necessarily so even for biconvex graphs.  Achieving low congestion in the monotone
setting is a delicate matter, and will be addressed in the following section.

It only remains to present a precise relationship between congestion
and mixing time, as given in~\cite[Cor.\,5.9]{Jerrum03}, based on Sinclair~\cite{Sin92}.
\begin{lemma}\label{lem:mixvscongestion}
The mixing time of a symmetric Markov chain with uniform stationary distribution
is bounded by
$$
\tmix(\epsilon)\,\leq\, 2\varrho\thsp (\ln|\Omega|+2\ln\epsilon^{-1}),
$$
where $\varrho$ is the congestion with respect to any set of canonical paths.
\end{lemma}
Technically, this result applies to a so-called ``lazy'' version of the
Markov chain, in which which a fair coin is flipped at each step. A transition is made
if the coin comes up ``heads'', and no transition is made if the coin comes up ``tails''.
In our case, the loop probabilities, which are at least $1/n$, are sufficient
to avoid the introduction of the lazy version. Then the factor~2 could be saved in the mixing time bound,
but we will not use this refinement.

\subsection{Construction of canonical paths}
\label{sec:pathconstruct}

Our goal is to construct canonical paths between arbitrary pairs $X$, $Y$ of perfect matchings
in~$G$. In general, $(V,X\cup Y)$ is a subgraph of $G$ that is composed of alternating
cycles $C_1\cup\cdots\cup C_s=X\oplus Y$, and isolated edges $X\cap Y$.
Assume that these cycles are ordered deterministically in some way.  (For example, order the cycles according
to the smallest unprimed vertex in each cycle.)  Then we may switch each of the cycles in order, using
the procedure we will describe below.  The isolated edges can be left untouched. Thus, it is sufficient
to construct the canonical path for a single alternating cycle.

In fact, we may specialise the canonical path construction even further.
Since \textsc{Monotone} is a hereditary class, if $H$ is any alternating cycle in $G$,
it is a Hamilton cycle in a smaller monotone graph $G[V(H)]$.
Thus we will assume that $G[V(H)]=G$ in the remainder of this section, and let $H$ be the (Hamilton) cycle with vertices
$(u_1,v'_1,u_2,v'_2,\ldots,u_n,v'_n)$. So our initial and final matchings are as follows:
\begin{align*}
X&=\big\{(u_1,v'_1),(u_2,v'_2),\ldots,(u_n,v'_n)\big\}\\
\noalign{and}
Y&=\big\{(u_2,v'_1),(u_3,v'_2),\ldots,(u_1,v'_n)\big\},
\end{align*}
where we will choose $u_1=n$ as the initial vertex of the cycle.

With each pair $(u,v')\in[n]\times[n]'$ (which may or may not correspond to
an edge of~$G$) we associate a point $p=(v,n-u+1)$ in $\rset^2$.
In particular, we associate the edges in~$X$ with the points
$\{p_i=(v_i,n-u_i+1):i\in[n]\}$ and those in~$Y$ with
points $\{q_i=(v_i,n-u_{i+1}+1):i\in[n]\}$ (interpreting $u_{n+1}$ as $u_1$).
We can think of this mapping as assigning Cartesian coordinates to the entries of the adjacency
matrix $A(G)$ in such a way that the pattern of entries plotted in $\rset^2$
looks exactly like the pattern of entries in the matrix $A(G)$ as conventionally written:
the $x$ coordinate increases with increasing column number, and the $y$-coordinate
decreases with increasing row number.
Denote the $x$- and $y$-coordinates of point $p\in\rset^2$ by $x(p)$ and~$y(p)$, so that
$p=(x(p),y(p))$.

Let $P=\{p_1,p_2,\ldots,p_n\}\cup\{q_1,q_2,\ldots,q_n\}\subset[n]^2$.
The alternating (Hamilton) cycle $X\cup Y$ corresponds to the cyclic sequence $(p_1,q_1,p_2,q_2,\ldots,
p_n,q_n)$ of points.  Join the adjacent points in this sequence
by line segments (omitting the return segment from $q_n$ to $p_1$)
to yield a continuous path~$\Pi$ from $p_1$ to~$q_n$.
This path consists of alternating horizontal and vertical segments.
By the choice $u_1=n$ for the initial vertex we have that $y(p_1)=y(q_n)=1$, i.e., that
the path begins and ends at the lowest point (corresponding to the final row of the matrix).
Somewhere along the way the path reaches the highest point $y(q_k)=y(p_{k+1})=n$ (corresponding
to the first row of the matrix).

The following lemma is inspired by the ``mountain climbing problem'' (see, for example~\cite{Tucker95}), and is proved in the Appendix.
\begin{lemma}\label{lem:climbing}
Suppose $\Pi$ is as above.  There are continuous functions $\focusA,\focusB:[0,1]\to\Pi$ satisfying
$\focusA(0)=p_1$, $\focusA(1)=q_k$, $\focusB(0)=q_n$, $\focusB(1)=p_{k+1}$, and $y(\focusA(t))=y(\focusB(t))$ for
all $t\in[0,1]$.  Moreover the set of events
$$
T=\big\{t\in[0,1]:\focusA(t)\in P\text{ or }\focusB(t)\in P\big\}
$$
has cardinality at most $n^2$.
\end{lemma}
The lemma is tight in the sense that $|T|$ may be within $O(n)$ of the
claimed upper bound;  see Theorem~\ref{thm:lb}, noting that $n$ in that
theorem is larger than $n$ here by a factor~2.
Note that the trajectories of $\focusA(t)$ and $\focusB(t)$ will not, in general,
move in a uniform sense along~$\Pi$: it may be
necessary for either or both of $\focusA(t)$ and $\focusB(t)$ to retreat along~$\Pi$
for periods in order to make progress later.

Before entering into the details of the construction of canonical paths,
here is a high-level overview.  We regard the general points $[n]^2$, the special points $P$,
and the path~$\Pi$ as forming a {\it board}, on which we play out a game involving
$n$~{\it tokens}.  Movements of the tokens will mirror switches in the canonical path from $X$ to~$Y$.
Place the tokens initially on the points $p_1,p_2,\ldots, p_n$
that represent the initial matching~$X$.  Our aim is to move these $n$ tokens
to the points $q_1,q_2,\ldots,q_n$ in a manner consistent with switches in the
original graph.   Thus, at each step, we will take two tokens and move them to
new locations.  Note that the two tokens we choose form the endpoints of
the diagonal of some axis-aligned rectangle, say~$R$.  We remove
these tokens and replace them on the endpoints of the opposite diagonal of~$R$.
For this operation
to represent a valid switch in the original graph, it is necessary for the new
locations for the tokens to correspond to 1's in the matrix $A(G)$.  Ensuring that
this happens at every step requires the tokens to be moved in a very particular
order.  Describing this order is at the heart of the matter.

We think of each of $\focusA(t)$ and $\focusB(t)$ as being a {\it focus} [of activity].  As time passes from
from 0 to~1, and points $\focusA(t)$ and $\focusB(t)$ trace out their trajectories along~$\Pi$, we shall move
tokens in the neighbourhood of the foci according to certain rules. (Refer to
Figure~\ref{fig:caseI} for a pictorial description of such a rule.)
If we remove the points $\focusA(t)$ and $\focusB(t)$ from the path~$\Pi$, we separate $\Pi$
into three connected pieces.  Denote the points in $P$ lying in the middle piece by~$P_U=P_U(t)$
and the remaining points by $P_L=P_L(t)$.  Note that $P_L(t)\cup P_U(t)=P$,
except at events $t\in T$, when $\focusA(t)\in P$ or $\focusB(t)\in P$.
We shall ensure that tokens lying in $P_U$ are in their original locations
(i.e., points of the form~$p_i$), while
those in $P_L$ are in their final locations (i.e., points of the form~$q_i$).
As time passes, $P_L(t)$ tends to increase in size and $P_U(t)$ tends to decrease,
corresponding to tokens passing from the initial points $\{p_1,\ldots,p_n\}$
to the final points $\{q_1,\ldots,q_n\}$.

The arrangement of tokens on the board at time~$t$ (viewed as a subset of points in $[n]^2$)
will be called the {\it configuration\/} of tokens and denoted $\sigma=\sigma(t)$.
As we want $\sigma$ to correspond to a perfect matching in~$G$, we insist that it contain
one point from every row $y=i\in[n]$ of the board and one point from every column $x=j\in[n]$.
In what follows, it may be worth bearing in mind the basic underlying strategy,
which is to keep the tokens as far as possible on the points~$P$.  If this can
done effectively, it will be possible to provide an encoding of the current
state (in the sense described of the previous Section~\ref{sec:paths})
by forming a perfect matching~$W$
guided by the points of $P$ that are not in the current configuration~$\sigma$.

As the foci trace out their trajectories, there will be
periods (open time intervals) when both $\focusA(t)$ and $\focusB(t)$ are both
on (open) vertical segments, and during these periods $y(\focusA(t))$ and $y(\focusB(t))$ are
either both monotonically increasing or both monotonically decreasing;  call these {\it v-periods}.
The v-periods are separated by {\it h-periods}, closed time intervals
during which one of $\focusA(t)$ or~$\focusB(t)$ is stationary and the other is moving horizontally.
During v-periods, $\sigma(t)$ is constant and well defined, in a manner that
will be explained presently.  We do not examine the configuration during
h-periods, so its definition there is not important;  one could decree, e.g.,
that the configuration there is the same as during the previous v-period.

During a v-period, neither $\focusA(t)$ nor $\focusB(t)$ is a member of~$P$;
as a consequence, $P_L(t)\cup P_U(t)=P$, so that every point in~$P$
is assigned either to $P_L(t)$ or $P_U(t)$.
So assume that $\focusA(t)$ and $\focusB(t)$ are both
on (open) vertical segments.  For convenience we introduce a local naming scheme
around $\focusA(t)$ and $\focusB(t)$.  Denote by $a_1$ and $a_2$ the lower and upper ends
of the line segment containing~$\focusA(t)$.  Continue the numbering
$\,\ldots,a_0,a_1,a_2,a_3,\ldots\,$ along the path~$\Pi$ as far as necessary;
this convention provides a local numbering of some subsequence
of $p_1,q_1,\ldots,p_n,q_n$, or its reversal.  In a similar way, introduce names for the points
around $\focusB(t)$, with $b_1$ and $b_2$ being the lower and upper ends of the line
segment containing $\focusB(t)$.

Suppose $\sigma\subset[n]^2$ is a configuration of tokens.  A {\it hole-pair\/}
is a pair $H$ of adjacent points $H=\{p_i,q_i\}$ or $H=\{q_i,p_{i+1}\}$ of~$P$ such
that $\sigma\cap H=\emptyset$.  As the foci move, we attempt to maintain the
following Invariant~I, which is a conjunction of I1--I3:
\begin{enumerate}
\item[I1] $\{a_1,a_2\}$ is a hole-pair.
\item[I2] If $x(a_1)<x(b_1)$ then $\{b_2,b_3\}$ is a hole-pair;
otherwise $\{b_0,b_1\}$ is a hole-pair.
\item[I3] The are no hole-pairs beyond these two.
\end{enumerate}
A number of consequences follow from I1--I3:
(i)~$\sigma(t)$ is completely determined by $\focusA(t)$ and $\focusB(t)$,
(ii)~$|\sigma\cap P|=n-1$,
(iii)~$\sigma\cap P_L\subseteq\{q_1,q_2,\ldots,q_n\}$,
and (iv)~$\sigma\cap P_U\subseteq\{p_1,p_2,\ldots,p_n\}$.
(Starting from a hole-pair, and working around $\Pi$, successive
points in~$P$ are forced to be alternately in and out of~$\sigma$;  this demonstrates~(i),
and the other conclusions follow as byproducts of this argument.)
Invariant~I may fail after a token-switch, but when this happens it will be
immediately reinstated at the following switch, as we shall see.
There will generally be a single token (and exceptionally three)
lying outside $P$ (i.e., $|\sigma\setminus P|=1$ or exceptionally $|\sigma\setminus P|=3$).
The position of such a token(s) will be called a {\it dislocation}, and is
denoted $d$ (or $d'$ or~$d''$)
in Figures~\ref{fig:caseI}--\ref{fig:caseIV} below.

Initially $\sigma=\{p_1,\ldots, p_n\}$ so Invariant~I will certainly not be
satisfied at the outset.  A similar remark applies to the final configuration
$\sigma=\{q_1,\ldots,q_n\}$.  We will see how to finesse this issue later.  For the
time being, we assume that we are under way, that Invariant~I is satisfied,
and that $t$ is in a v-period, so that
$\focusA(t)$ and $\focusB(t)$ are on vertical line segments $(a_1,a_2)$ and $(b_1,b_2)$.
We assume first that $\focusA(t)$ (and hence $\focusB(t)$) is moving
upwards, i.e., that the rate of change of $y(\focusA(t))$ is positive.
(It transpires that the situation when $\focusA(t)$ is moving downwards
creates little extra work, as it can be handled by symmetry, specifically
by rotating the board through an angle~$\pi$.  More detail will be provided
immediately after we have dealt with the focus-moving-upwards situation.)
Depending on the ordering
of $y(a_2)$ and $y(b_2)$ one of two events occurs first:  either $y(\focusA(t))=y(a_2)$
(an $\focusA$-event) or $y(\focusB(t))=y(b_2)$ (a $\focusB$-event).
We consider the situation just before this event and just
after, and, in particular, what action needs to be taken to maintain the invariant.
(By ``just after'' here, we mean ``at a point in time just after the start of the
next v-period''.)

We proceed by case analysis.  There are eight cases I--IV and I*--IV*.
First we split on whether $x(a_1)<x(b_1)$ (Cases I and II) or
$x(a_1)>x(b_1)$ (Cases III and IV),
i.e., on the horizontal relationship between $\focusA(t)$ and $\focusB(t)$ before the event.
Then we split on whether $y(a_2)<y(b_2)$ (Cases I and III) or $y(a_2)>y(b_2)$
(Cases II and IV), i.e., on whether it is an $\focusA$-event or a $\focusB$-event that occurs.
Finally we split on whether $\focusA(t)$ and $\focusB(t)$ have the same horizontal
relationship after the event (unstarred cases) or opposite (starred cases).
Note that the eight Cases I--IV* are exhaustive.
We consider the eight cases in turn below.  Cases I and I* illustrate most of the
issues, so we treat them in more detail;  the remaining cases will leave more
for the reader.  Some shorthand terminology will be useful in describing these
cases.  Let $p,q\in\rset^2$ be points.  We say that $p$ is {\it to the left of}~$q$
(and $q$ is {\it to the right of}~$p$) if $y(p)=y(q)$ and $x(p)<x(q)$.
We say that $p$ is {\it above}~$q$ (and $q$ is {\it below}~$p$)
if $x(p)=x(q)$ and $y(p)>y(q)$.

\tikzset{tok/.style={circle,draw,inner sep=1pt,fill=lightgrey,minimum size= 6mm}}
\tikzset{notok/.style={circle,draw,inner sep=1pt,fill=none,minimum size=6mm}}
\tikzset{ghost/.style={circle,dotted,draw,inner sep=1pt,fill=none,minimum size=6mm}}
\tikzset{focus/.style={circle,draw,inner sep=1pt,fill=black,minimum size=1mm}}
\tikzset{every picture/.style={line width=0.8pt}}
\tikzset{empty/.style={rectangle,draw=none,fill=none}}


\begin{figure}[t]
\begin{center}
\scalebox{0.9}{\begin{tikzpicture}
\draw (1,0)  node[notok] (a1) {$a_1$} ++ (0,1.5) node [notok] (a2) {$a_2$} ++ (-1,0) node [tok] (a3) {$a_3$} ++ (0,2) node [notok] (a4) {$a_4$};
\draw (a1) -- (a2) -- (a3) -- (a4);
\draw (2,0.5) node[tok] (b1) {$b_1$} ++ (0,2)  node[notok] (b2) {$b_2$} ++ (1,0) node [notok] (b3) {$b_3$};
\draw (b1) -- (b2) -- (b3);
\draw (1,2.5) node[tok] {$d$};
\draw [dashdotted] (-0.5,1) -- (3.5,1);
\draw [->] (3.7,0.8) -- (3.7,1.2);
\draw [->] (-0.7,0.8) -- (-0.7,1.2);
\draw (1,1) node[focus] {};
\draw (0.7,0.8) node[empty] {$\focusA$};
\draw (2,1) node[focus] {};
\draw (2.3,1.25) node[empty] {$\focusB$};
\draw (4.5,1.5) node[empty] {$\rightarrow$};

\begin{scope}[shift={(6,0)}]
\draw (1,0)  node[notok] (a1) {$a_1$} ++ (0,1.5) node [tok] (a2) {$a_2$} ++ (-1,0) node [notok] (a3) {$a_3$} ++ (0,2) node [notok] (a4) {$a_4$};
\draw (a1) -- (a2) -- (a3) -- (a4);
\draw (2,0.5) node[tok] (b1) {$b_1$} ++ (0,2)  node[notok] (b2) {$b_2$} ++ (1,0) node [notok] (b3) {$b_3$};
\draw (b1) -- (b2) -- (b3);
\draw (0,2.5) node[tok] {$d'$};
\draw [dashdotted] (-0.5,2) -- (3.5,2);
\draw (0,2) node[focus] {};
\draw (0.4,2.2) node[empty] {$\focusA$};
\draw (2,2) node[focus] {};
\draw (2.2,1.75) node[empty] {$\focusB$};
\end{scope}

\end{tikzpicture}}\\[1ex]

Case I applies when $x(a_1)<x(b_1)$, $y(a_2)<y(b_2)$ and $x(a_3)<x(b_1)$.

\bigskip

\scalebox{0.9}{\begin{tikzpicture}
\draw (0,-0.5)  node[notok] (a1) {$a_1$} ++ (0,1.5) node [notok] (a2) {$a_2$} ++ (3,0) node [tok] (a3) {$a_3$} ++ (0,1.5) node [notok] (a4) {$a_4$};
\draw (a1) -- (a2) -- (a3) -- (a4);
\draw (2,0) node[notok] (b0) {$b_0$} ++ (-1,0)  node[tok] (b1) {$b_1$} ++ (0,2) node [notok] (b2) {$b_2$} ++ (1.5,0) node [notok] (b3) {$b_3$};
\draw (b0) -- (b1) -- (b2) -- (b3);
\draw (0,2) node[tok] {$d$};
\draw [dashdotted] (-0.5,0.5) -- (3.5,0.5);
\draw [->] (3.7,0.3) -- (3.7,0.7);
\draw [->] (-0.7,0.3) -- (-0.7,0.7);
\draw (0,0.5) node[focus] {};
\draw (-0.3,0.3) node[empty] {$\focusA$};
\draw (1,0.5) node[focus] {};
\draw (1.3,0.75) node[empty] {$\focusB$};
\draw (4.5,1) node[empty] {$\rightarrow$};

\begin{scope}[shift={(5.5,0)}]
\draw (0,-0.5)  node[notok] (a1) {$a_1$} ++ (0,1.5) node [notok] (a2) {$a_2$} ++ (3,0) node [notok] (a3) {$a_3$} ++ (0,1.5) node [notok] (a4) {$a_4$};
\draw (a1) -- (a2) -- (a3) -- (a4);
\draw (2,0) node[notok] (b0) {$b_0$} ++ (-1,0)  node[notok] (b1) {$b_1$} ++ (0,2) node [notok] (b2) {$b_2$} ++ (1.5,0) node [notok] (b3) {$b_3$};
\draw (b0) -- (b1) -- (b2) -- (b3);
\draw (0,2) node[tok] {$d$};
\draw (1,1) node[tok] {$d''$};
\draw (3,0) node[tok] {$d'$};
\draw (0,0.5) node[focus] {};
\draw (-0.3,0.3) node[empty] {$\focusA$};
\draw (1,0.5) node[focus] {};
\draw (1.3,0.6) node[empty] {$\focusB$};
\draw (4,1) node[empty] {$\rightarrow$};
\end{scope}

\begin{scope}[shift={(11,0)}]
\draw (0,-0.5)  node[notok] (a1) {$a_1$} ++ (0,1.5) node [tok] (a2) {$a_2$} ++ (3,0) node [notok] (a3) {$a_3$} ++ (0,1.5) node [notok] (a4) {$a_4$};
\draw (a1) -- (a2) -- (a3) -- (a4);
\draw (2,0) node[notok] (b0) {$b_0$} ++ (-1,0)  node[notok] (b1) {$b_1$} ++ (0,2) node [tok] (b2) {$b_2$} ++ (1.5,0) node [notok] (b3) {$b_3$};
\draw (b0) -- (b1) -- (b2) -- (b3);
\draw (3,0) node[tok] {$d'$};
\draw [dashdotted] (-0.5,1.5) -- (3.5,1.5);
\draw (3,1.5) node[focus] {};
\draw (3.3,1.7) node[empty] {$\focusA$};
\draw (1,1.5) node[focus] {};
\draw (1.2,1.25) node[empty] {$\focusB$};
\end{scope}

\end{tikzpicture}}\\[1ex]

Case I* applies when $x(a_1)<x(b_1)$, $y(a_2)<y(b_2)$ and $x(a_3)>x(b_1)$.
\end{center}

\caption{Cases I and I*}
\label{fig:caseI}
\end{figure}

{\it Case I.}
This case is characterised by $x(a_1)<x(b_1)$, $y(a_2)<y(b_2)$ and $x(a_3)<x(b_1)$.
(Refer to Figure~\ref{fig:caseI}, which illustrates the layout of the board in
the vicinity of the foci.)
In this figure and the later ones the conventions are as follows.
The dotted-and-dashed line is $y=y(\focusA(t))\>(=y(\focusB(t)))$
and marks the current $y$-coordinate of the foci.
Points with tokens are grey/red and points
without tokens are white.  The ``before'' picture is to
the left and the ``after'' to the right.  The current location of the dislocation
is marked~$d$.  Note that $d$ must always be the intersection of the vertical
line through $a_1$ and $a_2$, and either the horizontal line through $b_0$ and $b_1$,
or through $b_2$ and $b_3$, as appropriate;  this is a consequence of the fact
that the board contains a token on every horizontal and vertical line.

The action to be taken in this case is to switch the tokens at $a_3$ and~$d$.
Recall that ``switch'', in the token interpretation,
means drawing the axis-aligned rectangle with
$a_3$ and~$d$ at the endpoints of one diagonal, and moving the tokens
to the endpoints of the other diagonal.  In this case, the tokens
end up at~$a_2$ and a new dislocation~$d'$.
The most demanding check to make is that the dislocation~$d'$ is
legally placed after the switch.  Recall that, to be legal, $d'$ must be
a point corresponding to a 1 in the matrix $A(G)$.  This part of the argument
will always appeal to monotonicity.  In this instance we see that the new
dislocation~$d'$
is above $a_3$ and to the left of~$b_2$;  since $a_3$ and $b_2$ both correspond
to 1s in the matrix then so does $d'$, by monotonicity (Lemma~\ref{med:lem008}).  Once this point has been
established, and noting that the new hole-pairs are $\{a_3,a_4\}$ and $\{b_2,b_3\}$,
it is easy to check that I1--I3 have been preserved.

A note of caution.
When checking the invariant, it is important to bear in mind that there are
no relations between points beyond the ones explicitly stated,
and the ones implied by the line segments being parallel to the axes.
Thus, in the first row of Figure~\ref{fig:caseI}, $d'$ is left of~$b_2$ since $x(d')=x(a_3)<x(b_1)=x(b_2)$.
However, $a_3$ (respectively, $b_3$) may be to the left or right of~$a_2$ (respectively~$b_2$),
and $a_4$ may be above or below~$a_3$.  (In case $a_4$ is below~$b_4$, the foci
will end up below where they have been drawn, but on the same vertical lines;
that is, $\alpha$ and~$a_3$ will be on the same vertical line, as will $\beta$ and~$b_1$.)
It is crucial not to make use of any unintended relationships between points.
It would have been impractical (and unnecessary) to display all the possible
combinations.


\begin{figure}[t]

\begin{center}
\scalebox{0.9}{\begin{tikzpicture}
\draw (0,0)  node[notok] (a1) {$a_1$} ++ (0,3) node [notok] (a2) {$a_2$};
\draw (a1) -- (a2);
\draw (2,-1.0) node[tok] (b1) {$b_1$} ++ (0,2)  node[notok] (b2) {$b_2$} ++ (1,0) node [notok] (b3) {$b_3$} ++ (0,1) node [tok] (b4) {$b_4$} ++ (-2,0) node [notok] (b5) {$b_5$};
\draw (b1) -- (b2) -- (b3) -- (b4) -- (b5);
\draw (0,1) node[tok] {$d$};
\draw (3,-0.5) node[ghost] (x) {$b_4$};
\draw [dotted] (b3) -- (x);
\draw [dashdotted] (-0.5,0.5) -- (3.5,0.5);
\draw [->] (3.7,0.3) -- (3.7,0.7);
\draw [->] (-0.7,0.3) -- (-0.7,0.7);
\draw (0,0.5) node[focus] {};
\draw (0.4,0.7) node[empty] {$\focusA$};
\draw (2,0.5) node[focus] {};
\draw (2.2,0.2) node[empty] {$\focusB$};
\draw (3,0.5) node[focus] {};
\draw (3.3,0.2) node[empty] {$\focusB'$};
\draw (4.5,1.5) node[empty] {$\rightarrow$};

\begin{scope}[shift={(6,0)}]
\draw (0,0)  node[notok] (a1) {$a_1$} ++ (0,3) node [notok] (a2) {$a_2$};
\draw (a1) -- (a2);
\draw (2,-0.5) node[tok] (b1) {$b_1$} ++ (0,1.5)  node[notok] (b2) {$b_2$} ++ (1,0) node [tok] (b3) {$b_3$} ++ (0,1) node [notok] (b4) {$b_4$} ++ (-2,0) node [notok] (b5) {$b_5$};
\draw (b1) -- (b2) -- (b3) -- (b4) -- (b5);
\draw (0,2) node[tok] {$d'$};
\draw [dashdotted] (-0.5,1.5) -- (3.5,1.5);
\draw (0,1.5) node[focus] {};
\draw (0.3,1.3) node[empty] {$\focusA$};
\draw (3,1.5) node[focus] {};
\draw (3.4,1.75) node[empty] {$\focusB$};
\end{scope}
\end{tikzpicture}}\\[1ex]

Case II applies when $x(a_1)<x(b_1)$, $y(a_2)>y(b_2)$ and $x(a_1)<x(b_3)$.
Only perform this switch if $y(b_4)>y(b_3)$; otherwise, just move focus $\focusB$ to $\focusB'$.

\bigskip

\scalebox{0.9}{\begin{tikzpicture}
\draw (2,0)  node[notok] (a1) {$a_1$} ++ (0,3) node [notok] (a2) {$a_2$};
\draw (a1) -- (a2);
\draw (3,0.5) node[tok] (b1) {$b_1$} ++ (0,1.5)  node[notok] (b2) {$b_2$} ++ (-2,0) node [notok] (b3) {$b_3$} ++ (0,-1) node [tok] (b4) {$b_4$} ++ (-1,0) node [notok] (b5) {$b_5$};
\draw (b1) -- (b2) -- (b3) -- (b4) -- (b5);
\draw (2,2) node[tok] {$d$};
\draw (1,3.5) node[ghost] (x) {$b_4$};
\draw [dotted] (b3) -- (x);
\draw [dashdotted] (-0.5,1.5) -- (3.5,1.5);
\draw (2,1.5) node[focus] {};
\draw (1.8,1.3) node[empty] {$\focusA$};
\draw (2,2.5) node[focus] {};
\draw (2.3,2.6) node[empty] {$\focusA'$};
\draw (3,1.5) node[focus] {};
\draw (3.2,1.25) node[empty] {$\focusB$};
\draw (1,2.5) node[focus] {};
\draw (1.3,2.6) node[empty] {$\focusB'$};
\draw (4.5,1.5) node[empty] {$\rightarrow$};

\begin{scope}[shift={(6,0)}]
\draw (2,0)  node[notok] (a1) {$a_1$} ++ (0,3) node [notok] (a2) {$a_2$};
\draw (a1) -- (a2);
\draw (3,0.5) node[tok] (b1) {$b_1$} ++ (0,1.5)  node[notok] (b2) {$b_2$} ++ (-2,0) node [tok] (b3) {$b_3$} ++ (0,-1) node [notok] (b4) {$b_4$} ++ (-1,0) node [notok] (b5) {$b_5$};
\draw (b1) -- (b2) -- (b3) -- (b4) -- (b5);
\draw (2,1) node[tok] {$d'$};
\draw [dashdotted] (-0.5,1.5) -- (3.5,1.5);
\draw (2,1.5) node[focus] {};
\draw (2.3,1.7) node[empty] {$\focusA$};
\draw (1,1.5) node[focus] {};
\draw (0.6,1.25) node[empty] {$\focusB$};
\end{scope}

\end{tikzpicture}}\\[1ex]

Case II* applies when $x(a_1)<x(b_1)$, $y(a_2)>y(b_2)$ and $x(a_1)>x(b_3)$.
Only perform this switch if $y(b_4)<y(b_3)$;  otherwise, just move the foci to $\focusA'$ and $\focusB'$.

\end{center}
\caption{Cases II and II*}
\label{fig:caseII}
\end{figure}

{\it Case I*.} This case is similar to the previous one, but with $x(a_3)>x(b_1)$ replacing
$x(a_3)<x(b_1)$.  The fact that $\focusA(t)$ starts to the left of $\focusB(t)$ and ends
up on the right creates an additional complication, whose resolution
necessitates an extra switch.

The first action to be taken is to switch the tokens at $a_3$ and $b_1$
creating two new dislocations at $d'$ and~$d''$.
This is legal because $d'$ is below $a_3$ and to the right of~$b_1$,
and because $d''$ is to the left of~$a_3$ and above $b_1$.
(Note that we use here both of the valid
modes of reasoning arising from monotonicity.)  At this point we have three
dislocations, which is a worst case.  Note that we have temporarily violated~I3,
but we will put this right at the following step, when we switch $d$ and~$d''$.
This switch requires no justification, as we are moving tokens to their official
locations in~$P$.  As an aside, we remark that $a_4$ may be below~$a_3$,
in which case $\focusA(t)$ and $\focusB(t)$ will end up lower than drawn in the figure.
In any case, it is routine to check that Invariant~I is satisfied after the second switch.

{\it Case II.}  This case is characterised by $x(a_1)<x(b_1)$, $y(a_2)>y(b_2)$ and $x(a_1)<x(b_3)$.
(Refer to Figure~\ref{fig:caseII}.)
A new feature here is that the switch is conditional, and is only made when $b_4$ is above~$b_3$.
In that case, we switch $b_4$ and $d$; this is legal because the new dislocation~$d'$ is above~$a_1$
and to the left of~$b_4$.

If $b_4$ is below $b_3$ (indicated by the ghostly $b_4$ in the figure)
then no switch is required to maintain the invariant;  all that is needed is
to move the focus~$\focusB$ rightwards to~$\focusB'$.

{\it Case II*.} This case is similar to the previous, but with $x(a_1)>x(b_3)$ replacing
$x(a_1)<x(b_3)$.  Again the switch is conditional, this time only being made when $b_4$
is below $b_3$.  In that case we switch $b_4$ and $d$;  this is legal because the
new dislocation~$d'$ is below~$a_2$ and to the right of~$b_4$.  Otherwise we just move the foci
to $\focusA'$ and~$\focusB'$.


\begin{figure}[t]
\begin{center}
\scalebox{0.9}{\begin{tikzpicture}
\draw (2,1)  node[notok] (a1) {$a_1$} ++ (0,1) node [notok] (a2) {$a_2$} ++ (1,0) node [tok] (a3) {$a_3$} ++ (0,1.5) node [notok] (a4) {$a_4$};
\draw (a1) -- (a2) -- (a3) -- (a4);
\draw (0,0) node[notok] (b0) {$b_0$} ++ (1,0)  node[notok] (b1) {$b_1$} ++ (0,3) node [tok] (b2) {$b_2$};
\draw (b0) -- (b1) -- (b2);
\draw (2,0) node[tok] {$d$};
\draw [dashdotted] (-0.5,1.5) -- (3.5,1.5);
\draw (1,1.5) node[focus] {};
\draw (0.8,1.25) node[empty] {$\focusB$};
\draw (2,1.5) node[focus] {};
\draw (1.6,1.3) node[empty] {$\focusA$};
\draw [->] (3.7,1.3) -- (3.7,1.7);
\draw [->] (-0.7,1.3) -- (-0.7,1.7);
\draw (4.5,1.5) node[empty] {$\rightarrow$};

\begin{scope}[shift={(6,0)}]
\draw (2,1)  node[notok] (a1) {$a_1$} ++ (0,1) node [tok] (a2) {$a_2$} ++ (1,0) node [notok] (a3) {$a_3$} ++ (0,1.5) node [notok] (a4) {$a_4$};
\draw (a1) -- (a2) -- (a3) -- (a4);
\draw (0,0) node[notok] (b0) {$b_0$} ++ (1,0)  node[notok] (b1) {$b_1$} ++ (0,3) node [tok] (b2) {$b_2$};
\draw (b0) -- (b1) -- (b2);
\draw (3,0) node[tok] {$d'$};
\draw [dashdotted] (-0.5,2.5) -- (3.5,2.5);
\draw (1,2.5) node[focus] {};
\draw (0.8,2.25) node[empty] {$\focusB$};
\draw (3,2.5) node[focus] {};
\draw (2.7,2.7) node[empty] {$\focusA$};

\end{scope}

\end{tikzpicture}}\\[1ex]

Case III applies when $x(a_1)>x(b_1)$, $y(a_2)<y(b_2)$ and $x(a_3)>x(b_1)$.

\bigskip

\scalebox{0.9}{\begin{tikzpicture}
\draw (3,1)  node[notok] (a1) {$a_1$} ++ (0,1) node [notok] (a2) {$a_2$} ++ (-2,0) node [tok] (a3) {$a_3$} ++ (0,1) node [notok] (a4) {$a_4$};
\draw (a1) -- (a2) -- (a3) -- (a4);
\draw (0.5,0) node[notok] (b0) {$b_0$} ++ (1.5,0)  node[notok] (b1) {$b_1$} ++ (0,4) node [tok] (b2) {$b_2$} ++ (-2,0) node [notok] (b3) {$b_3$};
\draw (b0) -- (b1) -- (b2) -- (b3);
\draw (3,0) node[tok] {$d$};
\draw [dashdotted] (-0.5,1.5) -- (3.5,1.5);
\draw [->] (3.7,1.3) -- (3.7,1.7);
\draw [->] (-0.7,1.3) -- (-0.7,1.7);
\draw (3,1.5) node[focus] {};
\draw (2.6,1.3) node[empty] {$\focusA$};
\draw (2,1.5) node[focus] {};
\draw (1.7,1.25) node[empty] {$\focusB$};
\draw (4.5,2) node[empty] {$\rightarrow$};

\begin{scope}[shift={(5.5,0)}]
\draw (3,1)  node[notok] (a1) {$a_1$} ++ (0,1) node [notok] (a2) {$a_2$} ++ (-2,0) node [notok] (a3) {$a_3$} ++ (0,1) node [notok] (a4) {$a_4$};
\draw (a1) -- (a2) -- (a3) -- (a4);
\draw (0.5,0) node[notok] (b0) {$b_0$} ++ (1.5,0)  node[notok] (b1) {$b_1$} ++ (0,4) node [notok] (b2) {$b_2$} ++ (-2,0) node [notok] (b3) {$b_3$};
\draw (b0) -- (b1) -- (b2) -- (b3);
\draw (3,0) node[tok] {$d$};
\draw (1,4) node[tok] {$d'$};
\draw (2,2) node[tok] {$d''$};
\draw (4,2) node[empty] {$\rightarrow$};
\end{scope}

\begin{scope}[shift={(11,0)}]
\draw (3,1)  node[notok] (a1) {$a_1$} ++ (0,1) node [tok] (a2) {$a_2$} ++ (-2,0) node [notok] (a3) {$a_3$} ++ (0,1) node [notok] (a4) {$a_4$};
\draw (a1) -- (a2) -- (a3) -- (a4);
\draw (0.5,0) node[notok] (b0) {$b_0$} ++ (1.5,0)  node [tok] (b1) {$b_1$} ++ (0,4) node [notok] (b2) {$b_2$} ++ (-2,0) node [notok] (b3) {$b_3$};
\draw (b0) -- (b1) -- (b2) -- (b3);
\draw (1,4) node[tok] {$d'$};
\draw [dashdotted] (-0.5,2.5) -- (3.5,2.5);
\draw (1,2.5) node[focus] {};
\draw (1.35,2.7) node[empty] {$\focusA$};
\draw (2,2.5) node[focus] {};
\draw (2.3,2.75) node[empty] {$\focusB$};

\end{scope}

\end{tikzpicture}}\\[1ex]

Case III* applies when $x(a_1)>x(b_1)$, $y(a_2)<y(b_2)$ and $x(a_3)<x(b_1)$.
\end{center}

\caption{Cases III and III*}
\label{fig:caseIII}
\end{figure}

{\it Case III.} This case is characterised by $x(a_1)>x(b_1)$, $y(a_2)<y(b_2)$ and $x(a_3)>x(b_1)$.
(Refer to Figure~\ref{fig:caseIII}.)  Here we switch $a_3$ and~$d$;  this is legal because
the new dislocation $d'$ is below $a_3$ and to the right of~$b_1$.

{\it Case III*.} This case is similar to the previous, but with $x(a_3)<x(b_1)$ replacing
$x(a_3)<x(b_1)$.  First we switch $a_3$ and $b_2$, creating new dislocations at $d'$ and~$d''$;
this is legal because the positions of $d'$ and $d''$ relative to $a_3$ and $b_2$.
We finish by switching $d$ and~$d''$.


\begin{figure}[t]

\begin{center}
\scalebox{0.9}{\begin{tikzpicture}
\draw (3,1)  node[notok] (a1) {$a_1$} ++ (0,2) node [notok] (a2) {$a_2$};
\draw (a1) -- (a2);
\draw (0,-0.5) node[notok] (b0) {$b_0$} ++ (1,0) node[notok] (b1) {$b_1$} ++ (0,2.5) node[tok] (b2) {$b_2$} ++ (1,0) node [notok] (b3) {$b_3$} ++ (0,-2) node [tok] (b4) {$b_4$} ++ (-1.5,0) node [notok] (b5) {$b_5$};
\draw (b0) -- (b1) -- (b2) -- (b3) -- (b4) -- (b5);
\draw (3,-0.5) node[tok] {$d$};
\draw [dashdotted] (-0.5,1.5) -- (3.5,1.5);
\draw (3,1.5) node[focus] {};
\draw (2.7,1.7) node[empty] {$\focusA$};
\draw (1,1.5) node[focus] {};
\draw (0.7,1.25) node[empty] {$\focusB$};
\draw [->] (3.7,1.3) -- (3.7,1.7);
\draw [->] (-0.7,1.3) -- (-0.7,1.7);
\draw (4.5,1.5) node[empty] {$\rightarrow$};

\begin{scope}[shift={(5.5,0)}]
\draw (3,1)  node[notok] (a1) {$a_1$} ++ (0,2) node [notok] (a2) {$a_2$};
\draw (a1) -- (a2);
\draw (0,-0.5) node[notok] (b0) {$b_0$} ++ (1,0) node[tok] (b1) {$b_1$} ++ (0,2.5) node[notok] (b2) {$b_2$} ++ (1,0) node [notok] (b3) {$b_3$} ++ (0,-2) node [tok] (b4) {$b_4$} ++ (-1.5,0) node [notok] (b5) {$b_5$};
\draw (b0) -- (b1) -- (b2) -- (b3) -- (b4) -- (b5);
\draw (3,2) node[tok] {$d''$};
\draw (2,3.5) node[ghost] (x) {$b_4$};
\draw [dotted] (b3) -- (x);
\draw (3,2.5) node[focus] {};
\draw (2.7,2.5) node[empty] {$\focusA'$};
\draw (2,2.5) node[focus] {};
\draw (1.7,2.5) node[empty] {$\focusB'$};

\draw (4,1.5) node[empty] {$\rightarrow$};
\end{scope}

\begin{scope}[shift={(11,0)}]
\draw (3,1)  node[notok] (a1) {$a_1$} ++ (0,2) node [notok] (a2) {$a_2$};
\draw (a1) -- (a2);
\draw (0,-0.5) node[notok] (b0) {$b_0$} ++ (1,0) node[tok] (b1) {$b_1$} ++ (0,2.5) node[notok] (b2) {$b_2$} ++ (1,0) node [tok] (b3) {$b_3$} ++ (0,-2) node [notok] (b4) {$b_4$} ++ (-1.5,0) node [notok] (b5) {$b_5$};
\draw (b0) -- (b1) -- (b2) -- (b3) -- (b4) -- (b5);
\draw (3,0) node[tok] {$d'$};
\draw [dashdotted] (-0.5,1.5) -- (3.5,1.5);
\draw (3,1.5) node[focus] {};
\draw (2.7,1.7) node[empty] {$\focusA$};
\draw (2,1.5) node[focus] {};
\draw (1.7,1.25) node[empty] {$\focusB$};
\end{scope}

\end{tikzpicture}}\\[1ex]

Case IV applies when $x(a_1)>x(b_1)$, $y(a_2)>y(b_2)$ and $x(a_1)>x(b_3)$.
Only perform the second switch if $y(b_4)<y(b_3)$; otherwise, stop after the first with the foci at $\focusA'$ and $\focusB'$.

\bigskip

\scalebox{0.9}{\begin{tikzpicture}
\draw (2,1)  node[notok] (a1) {$a_1$} ++ (0,3) node [notok] (a2) {$a_2$};
\draw (a1) -- (a2);
\draw (0,0) node[notok] (b0) {$b_0$} ++ (1,0) node[notok] (b1) {$b_1$} ++ (0,2)  node[tok] (b2) {$b_2$} ++ (2,0) node [notok] (b3) {$b_3$} ++ (0,1) node [tok] (b4) {$b_4$} ++ (-2.5,0) node [notok] (b5) {$b_5$};
\draw (b0) -- (b1) -- (b2) -- (b3) -- (b4) -- (b5);
\draw (2,0) node[tok] {$d$};
\draw [dashdotted] (-0.5,1.5) -- (3.5,1.5);
\draw (2,1.5) node[focus] {};
\draw (1.8,1.7) node[empty] {$\focusA$};
\draw (1,1.5) node[focus] {};
\draw (0.7,1.25) node[empty] {$\focusB$};
\draw [->] (3.7,1.3) -- (3.7,1.7);
\draw [->] (-0.7,1.3) -- (-0.7,1.7);
\draw (4.5,1.5) node[empty] {$\rightarrow$};

\begin{scope}[shift={(5.5,0)}]
\draw (2,1)  node[notok] (a1) {$a_1$} ++ (0,3) node [notok] (a2) {$a_2$};
\draw (a1) -- (a2);
\draw (0,0) node[notok] (b0) {$b_0$} ++ (1,0) node[tok] (b1) {$b_1$} ++ (0,2)  node[notok] (b2) {$b_2$} ++ (2,0) node [notok] (b3) {$b_3$} ++ (0,1) node [tok] (b4) {$b_4$} ++ (-2.5,0) node [notok] (b5) {$b_5$};
\draw (b0) -- (b1) -- (b2) -- (b3) -- (b4) -- (b5);
\draw (2,2) node[tok] {$d''$};
\draw (3,0.5) node[ghost] (x) {$b_4$};
\draw [dotted] (b3) -- (x);
\draw (2,1.5) node[focus] {};
\draw (1.7,1.5) node[empty] {$\focusA'$};
\draw (3,1.5) node[focus] {};
\draw (2.7,1.5) node[empty] {$\focusB'$};
\draw (4.5,1.5) node[empty] {$\rightarrow$};
\end{scope}

\begin{scope}[shift={(11,0)}]
\draw (2,1)  node[notok] (a1) {$a_1$} ++ (0,3) node [notok] (a2) {$a_2$};
\draw (a1) -- (a2);
\draw (0,0) node[notok] (b0) {$b_0$} ++ (1,0) node[tok] (b1) {$b_1$} ++ (0,2)  node[notok] (b2) {$b_2$} ++ (2,0) node [tok] (b3) {$b_3$} ++ (0,1) node [notok] (b4) {$b_4$} ++ (-2.5,0) node [notok] (b5) {$b_5$};
\draw (b0) -- (b1) -- (b2) -- (b3) -- (b4) -- (b5);
\draw (2,3) node[tok] {$d'$};
\draw [dashdotted] (-0.5,2.5) -- (3.5,2.5);
\draw (2,2.5) node[focus] {};
\draw (1.8,2.3) node[empty] {$\focusA$};
\draw (3,2.5) node[focus] {};
\draw (2.6,2.25) node[empty] {$\focusB$};
\end{scope}

\end{tikzpicture}}\\[1ex]

Case IV* applies when $x(a_1)>x(b_1)$, $y(a_2)>y(b_2)$ and $x(a_1)<x(b_3)$.
Only perform the second switch if $y(b_4)>y(b_3)$; otherwise, stop after the first with the foci at $\focusA'$ and $\focusB'$.

\end{center}
\caption{Cases IV and IV*}
\label{fig:caseIV}
\end{figure}

{\it Case IV.} This case is characterised by $x(a_1)>x(b_1)$, $y(a_2)>y(b_2)$ and $x(a_1)>x(b_3)$.
(Refer to Figure~\ref{fig:caseIV}.)  Here and in the final case, we have an obligatory switch
followed by a conditional one.  First we switch $b_2$ and $d$, which is legal because
the new dislocation~$d''$ is below $a_2$ and to the right of $b_2$.  If $b_4$ is above $b_3$
(indicated by the ghostly $b_4$ in the middle frame of the sequence), we stop at this
point, with the foci at $\focusA'$ and~$\focusB'$.  Otherwise ($b_4$ is below $b_3$), we switch $b_4$
and~$d''$, which is legal as the new dislocation~$d'$ is below $a_2$ and to the right of~$b_4$.

{\it Case IV*.} This case is similar to the previous, but with $x(a_1)<x(b_3)$ replacing
$x(a_1)>x(b_3)$.  First we switch $b_2$ and $d$, which is legal because
the new dislocation~$d''$ is below $a_2$ and to the right of $b_2$.  If $b_4$ is below $b_3$,
we stop at this point, with the foci at $\focusA'$ and~$\focusB'$.
Otherwise ($b_4$ is above $b_3$), we switch $b_4$
and~$d''$, which is legal as the new dislocation~$d'$ is above $a_1$ and to the left of~$b_4$.

\begin{figure}[th]

\begin{center}
\scalebox{0.9}{\begin{tikzpicture}
\begin{scope}[shift={(-4.5,0)}]
\draw (0,2) node[notok] (q1) {$q_1$} ++ (0,-1) node [tok] (p1) {$p_1$} ++ (1,0) node[notok] (qn) {$q_n$} ++ (0,2) node [tok] (pn) {$p_n$};
\draw (q1) -- (p1);
\draw (qn) -- (pn);
\draw (-0.4,1) node [empty,left] {$(h,1)$};
\end{scope}

\draw (0,2)  node[notok] (q1) {$q_1$} ++ (0,-2) node [notok] (p1') {} ++ (4,0) node[tok] (qn+1) {} ++ (0,1) node[notok] (pn+1) {} ++ (-3,0) node[notok] (qn) {$q_n$} ++ (0,2) node [tok] (pn) {$p_n$};
\draw (q1) -- (p1');
\draw (qn+1) -- (pn+1) -- (qn) -- (pn);
\draw (4.4,1) node [empty,right] {$p_{n+1}=(n+1,1)$};
\draw (4.4,0) node [empty,right] {$q_{n+1}=(n+1,0)$};
\draw (-0.4,0) node [empty,left] {$p_1'=(h,0)$};
\draw (0,1) node[tok] {$d$};
\draw [dashdotted] (-0.5,0.5) -- (4.5,0.5);
\draw (0,0.5) node[focus] {};
\draw (0.35,0.7) node[empty] {$\focusA$};
\draw (4,0.5) node[focus] {};
\draw (3.55,0.75) node[empty] {$\focusB$};

\end{tikzpicture}}

\end{center}
\caption{Starting the canonical path with the aid of extra ``virtual points''}
\label{fig:finesse}
\end{figure}
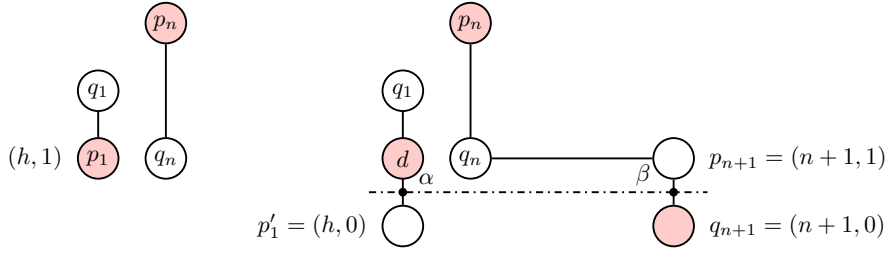

Cases~I--IV* exhaust all the possibilities with $\focusA(t)$ is moving upwards.
To deal with the situation when $\focusA(t)$ is moving downwards, we appeal to
symmetry.  The argument is as follows.  Rotate the board through an angle of~$\pi$,
an operation which is equivalent to the change of coordinates $x:=n-x+1$ and $y:=n-y+1$.
Note that, relative to the new coordinates, the configuration is still legal,
satisfying Invariant~I.  Note also that the underlying matrix is still monotonic,
since monotonicity of a matrix is preserved by rotation through~$\pi$.
However, $\focusA(t)$ is now moving upwards in the new coordinate system.
Now apply the switch or switches dictated by the relevant one of the Cases~I-IV*
described above.  Finally rotate the board back through an angle~$\pi$
to obtain a new configuration satisfying Invariant~I.

We now have to deal with the issue that the initial and final configurations
do not satisfy the invariant. (There are no hole-pairs.)  To get the procedure above going, we need to
create one horizontal and one vertical hole-pair.  We make space for this by
extending the board and adding three ``virtual points'' $p_1'$, $p_{n+1}$ and
$q_{n+1}$.  Suppose the coordinates of $p_1$ are $(h,1)$.  Delete the point
$p_1$ and add points $p_1'=(h,0)$, $p_{n+1}=(n+1,1)$ and $q_{n+1}=(n+1,0)$.
(In terms of the graph $G$, this corresponds to adding new vertices $n+1$
and $(n+1)'$ and new edges $(n+1,h')$, $(n,(n+1)')$ and $(n+1,(n+1)')$, together
with any further edges implied by monotonicity.)
Integrate these new points into the path $\Pi$ as illustrated in Figure~\ref{fig:finesse}.
Add a token to $q_{n+1}$ and leave the existing tokens in place (including
the one at the point formerly known as~$p_1$, which now becomes dislocation~$d$).
Place the foci $\focusA$ and $\focusB$ just below $d$ and $p_{n+1}$.  The invariant is now
satisfied, and the arrangement of tokens within the original extent of the board
is unchanged.  We can now start the canonical path construction as described earlier.
A similar construction (rotated through $\pi$) works to finish the path.

Two important points to note.  The new board is consistent with monotonicity.
(In this context, recall that the edge $(n,n')$ must be present in~$G$,
and hence the point $(n,1)$ is a valid position for a token.)
Moreover, even though $p_1$ has been placed to the left of $q_n$ in the figure,
the construction works equally well with $p_1$ to the right.

\subsection{Encoding and congestion}

In Lemma~\ref{lem:climbing}, reverse the role of $\focusA$ and $\focusB$, so that
$\focusA(0)=q_n$, $\focusA(1)=p_{k+1}$, $\focusB(0)=p_1$ and $\focusB(1)=q_k$.
Place the $n$~tokens initially on the points $\{q_1,\ldots,q_n\}$,
and denote the configuration at time~$t$ by $\sigma'(t)$.
Since the trajectories of $\focusA(t)$ and $\focusB(t)$ are oblivious of the
tokens, $P_L=P_L(t)$ and $P_U=P_U(t)$ are unchanged.
According to the invariant, the configuration~$\sigma'$
satisfies $\sigma'\cap P_U\subseteq\{q_1,\ldots,q_n\}$, $\sigma'\cap P_L\subseteq
\{p_1,\ldots,p_n\}$ and $|\sigma'\cap P|=n-1$.  At any legal time~$t$, then,
\begin{equation}\label{eq:complement}
\bigl|(\sigma(t)\cup\sigma'(t))\cap P\bigr|=2n-2.
\end{equation}

Consider a canonical path $X=Z_0\to Z_1\to \cdots \to Z_\ell=Y$ constructed by the
method of Section~\ref{sec:pathconstruct}.  Some of Cases I--IV* involved making two switches;
in these cases, call the middle configuration between the two switches
(and the corresponding perfect matching $Z_i$) {\it transitory}.
Note the configurations that are not transitory are of the form $\sigma(t)$ for
some~$t$, and these configurations satisfy Invariant~I.
If $Z_i$ is not transitory, consider a time $t$ at which configuration $\sigma(t)$
corresponds to~$Z_i$.  Then $\sigma'(t)$ is a near complement to $\sigma(t)$ with respect
to~$P$,
and its corresponding perfect matching $Z_i'$ is a near complement to $Z_i$
with respect to $X\cup Y$;  specifically, from equation~(\ref{eq:complement}),
$$\bigl|(Z_i\cup Z_i')\cap(X\cup Y)\bigr|=2n-2.$$

Suppose $(Z,Z^\dag)$ is a valid transition (switch) of the Markov chain.  In the usual
way, we wish to provide every canonical path through $(Z,Z^\dag)$ with a unique encoding,
and hence bound the total number of paths using $(Z,Z^\dag)$.  In fact, our encoding
will be an element of $\Omega\times[8n^2]$.  So suppose the transition
of interest occurs at $(Z_i,Z_{i+1})$ in a canonical path of the above form from $X$ to~$Y$.
We suppose first that $C=X\cup Y$ is a single cycle.
At least one of $Z_i$ and $Z_{i+1}$ is not transitory, say, $Z_i$.  Our encoding will
be the near complement $Z_i'$ together with some additional data.  Consider
$C'=Z_i\cup Z_i'$.  Noting that $|C'\setminus C|\leq2$, the first piece of extra data
we provide is the identity of the (at most) two edges in $C'\setminus C$ that need
to be deleted from~$C'$ in a bid to recover $C$. These edges cannot be either
the top or bottom horizontal edges, since these remain in $C'$ throughout.Thus there are at most
$\binom{2n-2}{2}+(2n-2)+1=2(n^2-n+1)\leq 2n^2$ possibilities for $C'\setminus C$.  Now
we need to add two edges, but there are only two ways this can be done so that
the result is a cycle.  Finally, we need to signal that the near complement was taken
with respect to $Z_i$ and not $Z_{i+1}$:  a further two possibilities.   This gives
us our encoding within the set $\Omega\times [8n^2]$.

In general, the symmetric difference of $X$ and $Y$ is a union of several cycles,
say $X\oplus Y= C_1\cup \cdots \cup C_k\cup\cdots\cup C_s$,
where the cycles are presented in the agreed order, and
$C_k$ is the cycle currently being switched using the procedure described
above.  To construct the perfect matching $W\in\Omega$ that
forms the main element of the encoding, proceed as follows.
Except on the cycle~$C_k$, the matching~$W$ is given by the symmetric difference
of $X$, $Y$ and~$Z_i$.  Thus, denoting the vertex set of~$C_k$ by~$V_k$, and the
complement $V\setminus V_k$ by~$\Vbar$, we have
$W[\,\Vbar\,]=(X\oplus Y\oplus Z_i)[\,\Vbar\,]$.
On $C_k$ itself, $W[V_k]$ is given by the single-cycle construction
from Section~\ref{sec:pathconstruct}.  The additional data is encoded by an integer in
$[8n^2]$ as before.  Note that we can reconstruct the common edges
$X\cap Y$ of the initial and final states,
and all the cycles of $X\oplus Y$ except $C_k$, by examining $Z_i\cap W$ and
$Z_i\oplus W$.  Knowing
that $C_1,\ldots, C_{k-1}$ have been switched and $C_{k+1},\ldots,C_s$ have not,
we may then recover $X[\,\Vbar\,]$ and $Y[\,\Vbar\,]$.
Finally, we recover $X[V_k]$ and $Y[V_k]$ as in the single-cycle case.
Summarising, there are at most
$8n^2|\Omega|$ canonical paths using any given transition.

We have all the quantities needed for the calculation of the congestion~$\varrho$.
From the definition of the switch chain, $\tprob(Z,Z^\dag)=2n^{-2}$.
From Lemma~\ref{lem:climbing}, the maximum length of a canonical path is $n^2$.
Substituting these values into (\ref{eq:congestion}) yields
$$
\varrho \leq  |\Omega|^{-1} (n^2/2)(8n^2)\,|\Omega|\,n^2 = 4n^6.
$$
Then, from Lemma~\ref{lem:mixvscongestion}, noting that
the state space $\Omega$ has cardinality at most $n!$, we obtain the
sought for bound on mixing time.
\begin{theorem}\label{thm:mixingtime}
  The switch Markov chain has mixing time $\tau(\varepsilon)< \,
  8n^6(n\ln n + 2\ln \varepsilon^{-1})=O(n^7\log n)$
  for any graph $G=([n]\cup[n]',E)$ in the class \textsc{Monotone}.
\end{theorem}
The big-O estimate comes from the observation that, if we were to set $\varepsilon$ smaller than $1/n!$, say,
we could not distinguish between the output of the chain and the uniform distribution within polynomial time.
By comparison with Theorem~\ref{thm:mixingtime}, the algorithm of \cite{BeStVV08} has running time $O(n^7\log^4 n)$, with no bound given on the implied constant. It may be possible to improve the analysis here, but it is highly unlikely that we could match the $O(n^2\log n)$ bound conjectured by \DGH~\cite{DiGrHo01}. Currently, this has  not even been established for chain graphs.

\section{Conclusions}\label{sec:conclusions}

The hierarchy of hereditary graph classes studied in this paper provides
a framework for understanding the ergodicity and mixing time of the
switch chain for perfect matchings on bipartite graphs. Unfortunately, while we have identified the largest
class for ergodicity, it seems difficult to characterise the largest
hereditary class that supports polynomial time mixing. Indeed, it is not even clear
that such a class exists.

A simpler question would be to exhibit proper superclasses
of monotone graphs for which the switch chain mixes rapidly.
From the examples $\mathscr{G}_k$ given in
subsection \ref{sec:biconvex}, we see that such a superclass cannot
contain all biconvex graphs. More precisely, it cannot contain
$\mathscr{G}_k$ for all $k$. Thus we might consider the class which excludes
$\mathscr{G}_k$ for some $k$, and try to determine the largest $k$ for which this
class has rapid mixing. Note that the class excluding $\mathscr{G}_k$ also excludes
$\mathscr{G}_\ell$ for all $\ell>k$, since $\mathscr{G}_k$ is an induced subgraph
of $\mathscr{G}_\ell$.

The graph $\mathscr{G}_2$ is monotone. Therefore the class of
$\mathscr{G}_2$-free biconvex graphs, that is, biconvex graphs that
have no induced subgraph isomorphic to $\mathscr{G}_2$, is
not a superclass of monotone graphs. However, for every fixed $k > 2$
the class of all $\mathscr{G}_k$-free biconvex graphs does contain all
monotone graphs. This can be proved by observing inclusions among the forbidden
subgraphs in Figs.~\ref{fig:Tucker} and~\ref{med:fig004}.
Unfortunately, we know very little about the structure
of the graphs in these classes, let alone whether their structure
might support rapid mixing.\vspace{-1mm}

\appendix
\section*{APPENDIX}
\setcounter{section}{0}

\section{Proof of Lemma~\ref{lem:climbing}}
\subsection{The parallel mountain climbers problem}
Let $P=(p_1,p_2,\ldots,p_n)$ be an ordered sequence of points $p_i=(x_i,y_i)\in\rset^2$ such that $y_1=y_n$ and $y_i\geq y_1$ ($i\in [n]$). The line segments $L_i=[p_i,p_{i+1}]$ ($i\in[n-1]$) will be called \emph{slopes}. We will assume that $y_i\neq y_j$ for $i\neq j$, unless $\{i,j\}=\{1,n\}$. Otherwise, we replace $y_i$ by $y_i+\varepsilon^i$ ($1<i<n$), for small enough $\varepsilon>0$.  (It is not really necessary to remove vertical slopes, since we will see that they can be identically to the others.) Similarly, we can avoid vertical slopes by replacing the $x_i$ by $x_i+\delta^i$ for some small $0<\delta<\varepsilon$.  The \emph{range} $\Pi$ is then the point set $\bigcup_{i=1}^{n-1}L_i \subset\rset^2$. The set $P\subset \Pi$ will be called the \emph{nodes} of the range. Observe that $\Pi$ is a connected one-dimensional manifold with boundary $\{p_1,p_n\}$, in which the nodes of $P$ appear in order.

The \emph{height} $y:\Pi\to\rset$ ($p\in \Pi$) is the piecewise linear function defined by interpolation on the $y_i$ $(i\in[n]$). Thus, if $p\in L_i$ is such that $p=\lambda p_i+(1-\lambda)p_{i+1}$, then $y(p)=\lambda y_i+(1-\lambda)y_{i+1}$.  A node $p_i$ will be called a \emph{peak} if it is a local maximum of $y$, a \emph{valley} if it is a local minimum of $y$, and \emph{monotone} otherwise. The \emph{summit} $p_s$ ($s\in[n]$) of $\Pi$ is the highest peak, so $y_s=\max\{y(p):p\in \Pi\}$.

We will say that two ranges $\Pi,\thsp\Pi'$ are \emph{isomorphic} if there is bijection between $P$ and $P'$ such that $y'_i<y'_j$ if and only if $y_i<y_j$. The bijection is extended linearly to the $L_i$. If $p=\lambda p_i + (1-\lambda) p_{i+1}\in L_i$, then $p'=\lambda p'_i + (1-\lambda) p'_{i+1}\in L'_i$. Hence, for any range $\Pi$, there is an isomorphic range $\Pi'$ such  that $p_i=(i,y_i-y_1)$ ($i\in[n]$). For illustration, we will generally use this \emph{standard presentation} of $\Pi$. Then, for example, we will say a point $v_1\in \Pi$ is to the left or right of a point $v_2\in \Pi$ if it is true in the standard presentation.

An example range with $n=11$ is shown below, first in the standard presentation, and then in an isomorphic presentation. Here $p_3,p_5,p_7,p_{10}$ are peaks, $p_1,p_4,p_6,p_9,p_{11}$ are valleys,
and $p_2,p_8$ are monotone. The summit is $p_7$.
\tikzset{every node/.style=empty}
\begin{figure}[ht]
\centerline{
\scalebox{0.75}{\begin{tikzpicture}[line width=1pt,xscale=0.75,yscale=0.5,font=\large]
\draw[color=gray!20!white,step=2cm] (0,0) grid (10,10) ;
\draw (0,0) node[snode](a){} (1,2.5)node[snode](b){}  (2,8)node[snode](c){}  (3,5)node[snode](d){} (4,6)node[snode](e){}  (5,3)node[snode](f){}  (6,10)node[snode](g){} (7,4)node[snode](h){}  (8,2)node[snode](i){}  (9,7)node[snode](j){}  (10,0)node[snode](k){};
\draw (a)--(b)--(c)--(d)--(e)--(f)--(g)--(h)--(i)--(j)--(k);
\draw (a) node[below]{$p_1$} (b) node[right]{$p_2$} (c) node[above]{$p_3$} (d) node[below]{$p_4$}
(e) node[above]{$p_5$} (f)  node[below]{$p_6$} (g) node[above]{$p_7$} (h) node[right]{$p_8$}
(i) node[below]{$p_9$} (j) node[above]{$p_{10}$} (k)  node[below]{$p_{11}$};
\draw[dotted] (h)--(4.6,4) (f)--(7.54,3) ;
\draw (5,-2) node {\large Range $\Pi$ in standard presentation};
\end{tikzpicture}}\hspace{2cm}
\scalebox{0.75}{\begin{tikzpicture}[line width=1pt,xscale=0.75,yscale=0.5,font=\large]
\draw[color=gray!20!white,step=2cm] (0,0) grid (10,10) ;
\draw (3,0) node[snode](a){} (7,2.5)node[snode](b){}  (8,8)node[snode](c){}  (8,5)node[snode](d){} (4,6)node[snode](e){}  (6,3)node[snode](f){}  (5,10)node[snode](g){} (1,4.25)node[snode](h){}  (3,2)node[snode](i){}  (9,7)node[snode](j){}  (8,0)node[snode](k){};
\draw (a)--(b)--(c)--(d)--(e)--(f)--(g)--(h)--(i)--(j)--(k);
\draw (a) node[below]{$p_1$} (b) node[right]{$p_2$} (c) node[above]{$p_3$} (d) node[below]{$p_4$}
(e) node[above]{$p_5$} (f)  node[below]{$p_6$} (g) node[above]{$p_7$} (h) node[left]{$p_8$}
(i) node[below]{$p_9$} (j) node[above]{$p_{10}$} (k)  node[below]{$p_{11}$};
\draw[dotted] (h)--(5.1,4.25) (f)--(2.1,3) ;
\draw (5,-2) node {\large A range isomorphic to $\Pi$};
\end{tikzpicture}}
}
\caption{Example}\label{fig:range}
\end{figure}
We wish to show the existence of continuous functions $\alpha,\beta:[0,1]\to \Pi$, with $\alpha(0)=p_1$, $\beta(0)=p_n$, $\alpha(1)=\beta(1)=p_s$, such that $y(\alpha(t))=y(\beta(t))$ for all $t\in[0,1]$. An \emph{event} will be any value of $t$ such that $\alpha(t)\in P$ or $\beta(t)\in P$, and we will let $T$ be the set of events. Our assumption that $y_i=y_j$ only if $i=j$ or $\{i,j\}=\{1,n\}$ implies that, for any $p_i$ ($1<i< n$), $y(p)=y_i$ only if $p=p_i$ or $p$ is in the interior of some line segment $L_j$. This can be viewed as the intersection of a horizontal line through $p_i$ with the line segment $L_j$. If this intersection exists, we will denote it by $p_iL_j$. The set of such intersections determines all potential events. It is a superset of $T$, and is clearly preserved under isomorphism.

Note that, in the interval $(t_j,t_{j+1})$ between two successive events, we may take $\alpha,\beta$ to be functions defined by linear interpolation:
\[ \alpha(t)=\frac{(t_{j+1}-t)\alpha(t_j)+(t-t_j)\alpha(t_{j+1})}{t_{j+1}-t_j},\quad \beta(t)=\frac{(t_{j+1}-t)\beta(t_j)+(t-t_j)\beta(t_{j+1})}{t_{j+1}-t_j}.\]
Thus $\alpha,\beta$ may be taken as piecewise linear functions. These functions clearly transform in an obvious way, under isomorphism. They are completely determined by the set $T$.  Hence we wish to determine a tight upper bound on $|T|$.

This is usually posed as the (discrete parallel) \emph{mountain climbers problem}, and we will use its terminology here. (See, for example,~\cite{GoPaYa89,Pak10,Tucker95,West01}.) For example, we refer to the parameter $t$ as ``time''. Two climbers, starting at $p_1$ and $p_n$, attempt to meet by moving continuously on the slopes of the standard presentation of the range $\Pi$, so that they are always at the same height. Note that one climber must always remain to the left of the summit, and the other to the right, since if either reaches the summit, so must the other. Thus, if they meet, they must meet at $p_s$. Can they always meet and, if so, how many events are required to ensure that they meet\vthsp?

Note that (a version of) this problem first arose in analysis, as a problem concerning continuous functions. (See, for example,~\cite{Keleti93,Lo99,Whitta96}.) This continuous version is somewhat different. There are pathological cases in which the climbers cannot meet and, even when they can, there may be no finite bound on the number of events.

There appears to be no worst-case analysis of the discrete version of this problem in the literature, although~\cite{GoPaYa89} gave an $O(n^2)$ upper bound on the number of events. Since this problem is central to the canonical path construction in this paper, we give an analysis here.

\subsection{Upper bound}
The \emph{range graph}~\cite{Tucker95} $\RG(\Pi)$ has a vertex for each intersection $p_iL_j$, with $1 < i < s \leq j<n$ or $1 \leq j< s < i < n$. These are intersections from points to the left of the summit with slopes to the right, and vice versa, representing all possible events during the climb.  We also add vertices for the initial and final states $p_1p_n$ and $p_sp_s$. It is easy to see that this vertex set is the same for isomorphic ranges and can have cardinality at most $(n-1)(n-2)+2$, by simple counting. However, this bound can be improved.
\begin{lemma}
There are at most $(n-1)^2/4+1$ vertices in $\RG(\Pi)$.
\end{lemma}
\begin{proof}
Suppose there are $n_1$ slopes to the left of the summit, and $n_2$ to the right, so $n_1+n_2=n-1$.
Let $N(n_1,n_2)$ be the maximum number of vertices in $\RG(\Pi)$. We will show by induction that $N(n_1,n_2)\leq n_1n_2+1$. First, it is easy to see that $N(1,n_2)\leq n_2+1$.

Next, if any of the nodes from $p_2$ to $p_{s-1}$ is monotone, consider the range given by removing such a node. Its range graph will have at most $N(n_1-1,n_2)\leq(n_1-1)n_2+1$ vertices, by induction. Adding back the removed node gives rise to at most $n_2$ new vertices, as shown below (in the case that the removed node is $p_2$), so $N(n_1,n_2)\leq(n_1-1)n_2+1+n_2=n_1n_2+1$ vertices.\vspace{3mm}
\begin{figure}[ht]
{\centering
\scalebox{0.8}{
\begin{tikzpicture}[line width=1pt,xscale=0.75,yscale=0.5]
\draw[color=gray!40!white,step=2cm] (0,10)--(0,0)--(9,0)--(9,10)--cycle ;
\draw (0,0)--(2,6)--(3,3)--(4,10)--(5,2)--(6,9)--(7,5)--(8,7)--(9,0);
\draw (0,0)node[snode]{}(2,6)node[snode]{}(3,3)node[snode]{}(4,10)node[snode]{}(5,2)node[snode]{}
(6,9)node[snode]{}(7,5)node[snode]{}(8,7)node[snode]{}(9,0)node[snode]{} (0.75,4)node[snode]{};
\draw[dashed] (0,0)--(0.75,4)--(2,6);
\draw (4.67,-1.5) node{\large Adding a leftmost monotone node};
\end{tikzpicture}}
\hspace{2cm}
\scalebox{0.8}{
\begin{tikzpicture}[line width=1pt,xscale=0.75,yscale=0.5]
\tikzset{inode/.style={circle,draw,inner sep=0pt,fill=white,minimum size=1mm}}
\draw[color=gray!40!white,step=2cm] (0,10)--(0,0)--(9,0)--(9,10)--cycle ;
\draw (0,0)--(2,6)--(3,3)--(4,10) (7,5)--(8,7);
\draw[name path=A] (4,10) coordinate (A1)--(5,2) coordinate (A2);
\draw[name path=B] (A2)--(6,9) coordinate (A3) --(7,5) ;
\draw[name path=C] (8,7)  coordinate (A4)--(9,0)  coordinate (A5);
\draw (0,0)node[snode]{}(2,6)node[snode]{}(3,3)node[snode]{}(4,10)node[snode]{}(5,2)node[snode]{}
(6,9)node[snode]{}(7,5)node[snode]{}(8,7)node[snode]{}(9,0)node[snode]{} (0.75,4)node[snode]{};
\draw[dashed] (0,0)--(0.75,4)--(2,6);
\draw[densely dotted,name path=H] (0.75,4) coordinate (H1)--(8.4,4)  coordinate (H2) ;
\coordinate (I1) at (intersection of A1--A2 and H1--H2) ;
\coordinate (I2) at (intersection of A2--A3 and H1--H2) ;
\coordinate (I3) at (intersection of A4--A5 and H1--H2) ;
\draw (I1) node[inode] {}  (I2) node[inode] {}   (I3) node[inode] {}  ;
\draw (4.5,-1.5) node{\large Vertices arising from the new node};
\end{tikzpicture}}
}
\end{figure}
Finally, if there are no monotone nodes and hence $p_2$ is a peak,
consider the range given by removing both node $p_2$ and node $p_3$, which is
necessarily a valley. Its range graph will have $N(n_1-2,n_2)\leq(n_1-2)n_2+1$ vertices, using the induction. Adding back $p_2$ and $p_3$ gives the situation shown below.
\begin{figure}[h]
\begin{center}
\scalebox{0.7}{
\begin{tikzpicture}[line width=1pt,xscale=0.75,yscale=0.5]
\fill[color=gray!10!white] (0,8)--(10,8)--(10,4)--(0,4)--cycle ;
\draw (0,0)node[snode]{}(3,6)node[snode]{}(4,3)node[snode]{}(5,10)node[snode]{}(6,2)node[snode]{}
(7,9)node[snode]{}(8,5)node[snode]{}(9,7)node[snode]{}(10,0)node[snode]{} ;
\draw[color=gray!40!white,step=2cm] (0,10)--(0,0)--(10,0)--(10,10)--cycle ;
\draw (0,0)--(0.5,4) (2,4)--(3,6)--(4,3)--(5,10)--(6,2)--(7,9)--(8,5)--(9,7)--(10,0);
\draw[dashed] (0.5,4)--(1,8)--(2,4);
\draw[dotted] (0.5,4)--(2,4) ;
\draw (5,-1.5) node{\large Adding a new leftmost peak};
\end{tikzpicture}}
\hspace{2cm}
\scalebox{0.7}{
\begin{tikzpicture}[line width=1pt,xscale=0.75,yscale=0.5]
\tikzset{inode/.style={circle,draw,inner sep=0pt,fill=white,minimum size=1mm}}
\fill[color=gray!10!white] (0,8)--(10,8)--(10,4)--(0,4)--cycle ;
\draw[color=gray!40!white,step=2cm] (0,10)--(0,0)--(10,0)--(10,10)--cycle ;
\draw (7.4,8)--(8,5)--(9,7)--(9.4,4);
\draw (0.5,4)--(1,8) coordinate (5)--(2,4)  coordinate (6);
\draw (2,4)node[snode]{} (1,8)node[snode]{} (8,5)node[snode]{} (9,7)node[snode]{};
\draw[densely dotted] (1,8)--(7.4,8) (2,4)--(9.4,4) (0.875,7) node[inode](1){}--(9,7) coordinate (3)  (0.625,5) node[inode](2){}--(8,5) coordinate (4) ;
\draw (intersection of 5--6 and 1--3) node[inode] {} ;
\draw (intersection of 5--6 and 2--4) node[inode] {} ;
\draw (5,-1.5) node{\large Vertices arising from the $j$th path};
\end{tikzpicture}}\vspace{-5mm}
\end{center}
\end{figure}

The only new vertices will be in the window determined by the slope to the right of the new first peak,
i.e., between $p_2$ and~$p_3$.
The window partitions the range to the right of the summit into some number $r$ of disjoint paths. Suppose that the $j$th path has $n_{2,j}$ slopes ($j=1,\ldots,r$), so $\sum_{i=1}^rn_{2,j}\leq n_2$. Then it is easy to see that the $j$th path gives rise to exactly $2n_{2,j}$ new vertices. Thus the total number of vertices added is  $2\sum_{i=1}^rn_{2,j}\leq 2n_2$. Thus $N(n_1,n_2)\leq N(n_1-2,n_2)+2n_2\leq(n_1-2)n_2+1+2n_2=n_1n_2+1$, and
we need only observe that $n_1n_2$ has a maximum of $(n-1)^2/4$, when $n_1=n_2= (n-1)/2$.
\end{proof}
The edges of $\RG(\Pi)$ correspond to possible intervals between successive events. There is an edge between two vertices in $\RG(\Pi)$ if and only if there is a parallel move of the two climbers which would take them from one vertex to the other. Thus there is an edge between vertices $p_iL_j$ and $p_kL_r$ if either $k\in\{i-1,i+1\}$ and $j=r$, or $j\in\{k-1,k\}$ and $r\in\{i-1,i\}$.  For example, $p_6L_8$ is adjacent to $p_8L_5$ in the range graph of the range of Fig.~\ref{fig:range} above, as illustrated.

It is clear that each vertex, other than $p_1p_n$ and $p_sp_s$, is adjacent to exactly two other vertices, corresponding to moving left or right  on the range at a peak or valley, and up or down from a monotone node.
Thus all vertices in the graph have degree 2, apart from $p_1p_n$ and $p_sp_s$, which have degree 1, since it is only possible to move up from $p_1p_n$, or down from $p_sp_s$.
Hence the range graph comprises a unique path joining $p_1p_n$ to $p_sp_s$, and possibly some disjoint cycles.
The number of edges on this path is at most $(n-1)^2/4$. Note that the path from $p_1p_n$ to $p_sp_s$ in the range graph corresponds precisely to the sequence of events in the canonical path of Section~\ref{sec:pathconstruct}.

The full range graph for the example with $n=11$ in Fig.~\ref{fig:range} above has $22$ vertices, and comprises a path and a single cycle, with $14$ and 8 vertices, respectively:

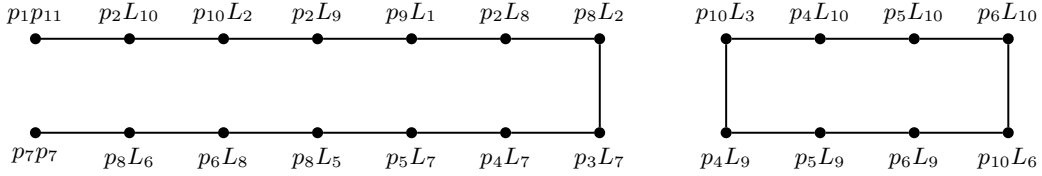
\begin{figure}[ht]
\centerline{\begin{tikzpicture}[line width=0.75pt,scale=1.25,font=\small]
\draw (0,0) node[cnode,label=above:{$p_1p_{11}$}] (1) {}
(1,0)  node[cnode,label=above:{$p_2L_{10}$}] (2) {}
(2,0)  node[cnode,label=above:{$p_{10}L_2$}] (3) {}
(3,0)  node[cnode,label=above:{$p_2L_9$}] (4) {}
(4,0) node[cnode,label=above:{$p_9L_1$}] (5) {}
(5,0) node[cnode,label=above:{$p_2L_8$}] (6) {}
(6,0) node[cnode,label=above:{$p_8L_2$}] (7) {}
(6,-1) node[cnode,label=below:{$p_3L_7$}] (8) {}
(5,-1) node[cnode,label=below:{$p_4L_7$}] (9) {}
(4,-1) node[cnode,label=below:{$p_5L_7$}] (10) {}
(3,-1) node[cnode,label=below:{$p_8L_5$}] (11) {}
(2,-1) node[cnode,label=below:{$p_6L_8$}] (12) {}
(1,-1) node[cnode,label=below:{$p_8L_6$}] (13) {}
(0,-1) node[cnode,label=below:{$p_7p_7$}] (14) {} ;
\draw (1)--(2)--(3)--(4)--(5)--(6)--(7)--(8)--(9)--(10)--(11)--(12)--(13)--(14) ;
\end{tikzpicture}\qquad
\begin{tikzpicture}[line width=0.75pt,scale=1.25,font=\small]
\draw (0,0) node[cnode,label=above:{$p_{10}L_3$}] (1) {}
(1,0) node[cnode,label=above:{$p_4L_{10}$}] (2) {}
(2,0) node[cnode,label=above:{$p_5L_{10}$}] (3) {}
(3,0) node[cnode,label=above:{$p_6L_{10}$}] (4) {}
(3,-1) node[cnode,label=below:{$p_{10}L_6$}] (5) {}
(2,-1) node[cnode,label=below:{$p_6L_9$}] (6) {}
(1,-1) node[cnode,label=below:{$p_5L_9$}] (7) {}
(0,-1) node[cnode,label=below:{$p_4L_9$}] (8) {} ;
\draw (1)--(2)--(3)--(4)--(5)--(6)--(7)--(8)--(1)  ;
\end{tikzpicture}}
\caption{The range graph for the example of Fig.~\ref{fig:range}}
\end{figure}

The upper bound on the number of vertices with $n=11$ is $10^2/4+1=26$, So the 22-vertex range graph has fewer vertices than the upper bound, and the 14-vertex path has considerably fewer.

This example illustrates that $\RG(\Pi)$ need not be connected, and may have fewer vertices than the upper bound. In fact, the range graph may have as few as $n-1$ vertices, for example when $s=2$. There cannot be fewer, since there must be at least one vertex for each of the $(n-3)$ nodes other than $p_1,\thsp p_s,\thsp p_n$, plus the two vertices $p_1p_n$, $p_sp_s$.

Note also that the effect on $\RG(\Pi)$ of removing the perturbations $\varepsilon^j$,  used to make heights unique, is to merge pairs of vertices $p_iL_{j-1}$ and $p_iL_j$ into a single vertex $p_ip_j$. The merged vertex may have degree 4 in the range graph, so the component of $\RG(\Pi)$ containing $p_1p_n$ and $p_sp_s$ may contain cycles. In this case, the path from $p_1p_n$ to $p_sp_s$ in $\RG(\Pi)$ may no longer be unique.
See~\cite{Tucker95} for an example.

\subsection{Alternative proof}

The proof given above is essentially that of Tucker~\cite{Tucker95}. Other authors (e.g.~\cite{GoPaYa89}) define the range graph to include all intersections $p_iL_j$, and even all ``intersections'' $p_ip_i$, but use essentially the same proof.  Unfortunately, the range graph proof does not provide a lower bound on the worst case, since the component structure of $\RG(\Pi)$ is opaque. To do this, we give an equally easy proof of the existence of a path from $p_1p_n$ and $p_sp_s$, by induction on $n$, which will lead to a lower bound in Section A.4 below.
If $n=3$, there is a single peak, and nothing to prove. Otherwise, we divide the problem into three stages. Assume, without loss of generality, that the second highest peak $p_r$ is to the left of the summit $p_s$, so $r<s$.  Let the deepest valley between $p_r$ and $p_s$ be $p_t$, so $r<t<s$. We construct the horizontal line $S_1$ from $p_r$  to the slope $L_s$ to the right of $p_s$,
and the horizontal line $S_2$ from $p_t$ to the first slope encountered to its right.
This produces three subproblems, climbing up from $p_1p_n$ to summit $S_1$, climbing down from $S_1$ to ``summit'' $S_2$,
and then climbing up from $S_2$ to the true summit $p_sp_s$. By ``climbing down'' we mean replacing the height function $y$ with~$-y$.

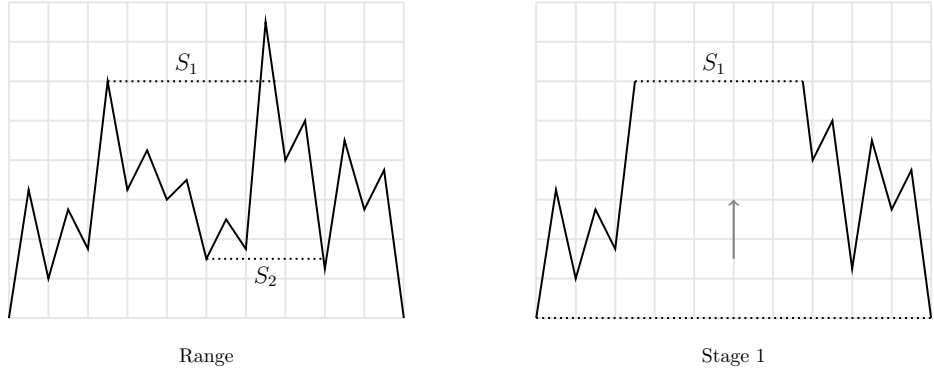
\begin{figure}[ht]
\centerline{
\scalebox{0.75}{\begin{tikzpicture}[line width=1pt,scale=0.35]
\draw[color=gray!20!white,step=2cm] (0,0) grid (20,16) ;
\draw (0,0)--(1,6.5)--(2,2)--(3,5.5)--(4,3.5)--(5,12)--(6,6.5)--(7,8.5)--(8,6)--(9,7)--(10,3)--(11,5)
--(12,3.5)--(13,15)--(14,8)--(15,10)--(16,2.5)--(17,9)--(18,5.5)--(19,7.5)--(20,0);
\draw[dotted] (5,12)--(13.5,12) (10,3)--(15.9,3);
\draw (10,-2) node {Range} (9,12)node[above]{\large$S_1$} (13,3)node[below]{\large$S_2$} ;
\end{tikzpicture}}
\hspace{1.5cm}
\scalebox{0.75}{\begin{tikzpicture}[line width=1pt,scale=0.35]
\draw[color=gray!20!white,step=2cm] (0,0) grid (20,16) ;
\draw (0,0)--(1,6.5)--(2,2)--(3,5.5)--(4,3.5)--(5,12) (13.5,12)--(14,8)--(15,10)--(16,2.5)--(17,9)--(18,5.5)--(19,7.5)--(20,0);
\draw[dotted] (5,12)--(13.5,12) (0,0)--(20,0);
\draw (10,3)edge[color=gray,->](10,6) (9,12)node[above]{\large$S_1$} (10,-2) node {Stage 1};
\end{tikzpicture}}
}\vspace{1ex}
\end{figure}

\begin{figure}[ht]
\centerline{
\scalebox{0.75}{\begin{tikzpicture}[line width=1pt,scale=0.35]
\draw[color=gray!20!white,step=2cm] (0,0) grid (20,16) ;
\draw (5,12)--(6,6.5)--(7,8.5)--(8,6)--(9,7)--(10,3)
(13.5,12)--(14,8)--(15,10)--(15.9,3);
\draw[dotted] (5,12)--(13.5,12) (10,3)--(15.9,3) ;
\draw (12,6)edge[color=gray,<-](12,9) (13,3)node[below]{\large$S_2$} (10,-2) node {Stage 2};
\end{tikzpicture}}
\hspace{1.5cm}
\scalebox{0.75}{\begin{tikzpicture}[line width=1pt,scale=0.35]
\draw[color=gray!20!white,step=2cm] (0,0) grid (20,16) ;
\draw (10,3)--(11,5)
--(12,3.5)--(13,15)--(14,8)--(15,10)--(15.9,3);
\draw[dotted] (10,3)--(15.9,3);
\draw (13.5,4)edge[color=gray,->](13.5,7) (10,-2) node {Stage 3};
\end{tikzpicture}}
}
\end{figure}

Each stage has a range with fewer nodes than the original problem, so each will have a path from the base to the summit in the corresponding range graph, by induction. Joining these paths together gives a path for the original problem. Note that the summits $S_1$ and $S_2$ in stages 1 and 2 must be shrunk to a point so that the ranges correspond to the definition, but this does not affect the path.

The uniqueness of the path is less easy to establish using this proof. However, we have already shown uniqueness in A.2 above, using the properties of the range graph, so we will not consider this further here.

\subsection{Lower bound}

For $n=2k+1\geq 3$, define the range $\Lambda_k$ with $p_1=(0,0)$, $p_{2k+1}=(2k,0)$,
$k$ peaks $p_i$ at $(i,\min\{k+i-1,3k-i+2\})$ ($i=2,4,\ldots,2k-2$), and $k-1$
valleys $p_i$ at $(i,\max\{k-i+2,i-k-1\})$ ($i=1,3,\ldots,2k-1$).
If $k$ is odd, $p_k=(k+1,2k)$ is the summit, and $p_{k+2}=(k+2,2k-1)$ is the second highest peak.
If $k$ is even, $p_{k+2}=(k+2,2k)$ is the summit, and $p_k=(k,2k-1)$ is the second highest peak.
Thus $s=k+1$ if $k$ is odd, and  $s=k+2$ if $k$ is even.
The case $n=21$ is shown below, together with the three ranges constructed in the inductive proof above.\\[1ex]
\begin{figure}[ht]
\centerline{%
\scalebox{0.7}{\begin{tikzpicture}[line width=1pt,xscale=0.4,yscale=0.33,font=\small]
\draw[color=gray!30!white,step=2cm] (0,0) grid (20,20) ;
\draw \foreach \x in {1,3,...,21}{(\x-1,-0.8) node{\x}  }
 \foreach \x in {0,2,...,20}{ (-0.8,\x) node{\x} };
\draw (0,0)--(1,11)--(2,9)--(3,13)--(4,7)--(5,15)--(6,5)--(7,17)--(8,3)--(9,19)--(10,1)
--(11,20)--(12,2)--(13,18)--(14,4)--(15,16)--(16,6)--(17,14)--(18,8)--(19,12)--(20,0);
\draw[dotted] (9,19)--(11.1,19) (10,1)--(19.9,1);
\draw (10,-3) node {\large Range $\Lambda_{10}$};
\end{tikzpicture}}
\hspace{1cm}
\scalebox{0.7}{\begin{tikzpicture}[line width=1pt,xscale=0.4,yscale=0.33,font=\small]
\draw[color=gray!30!white,step=2cm] (0,0) grid (20,20) ;
\draw \foreach \x in {1,3,...,21}{(\x-1,-0.8) node{\x}  }
 \foreach \x in {0,2,...,20}{ (-0.8,\x) node{\x} };
\draw (0,0)--(1,11)--(2,9)--(3,13)--(4,7)--(5,15)--(6,5)--(7,17)--(8,3)--(9,19)
(11.1,19)--(12,2)--(13,18)--(14,4)--(15,16)--(16,6)--(17,14)--(18,8)--(19,12)--(20,0);
\draw[dotted] (9,19)--(11.1,19) ;
\draw (10,-3) node {\large Stage 1};
\end{tikzpicture}}
}
\end{figure}

\begin{figure}[ht]
\centerline{
\scalebox{0.7}{\begin{tikzpicture}[line width=1pt,xscale=0.4,yscale=0.33,font=\small]
\draw[color=gray!30!white,step=2cm] (0,0) grid (20,20) ;
\draw \foreach \x in {1,3,...,21}{(\x-1,-0.8) node{\x}  }
 \foreach \x in {0,2,...,20}{ (-0.8,\x) node{\x} };
\draw (9,19)--(10,1)
(11.1,19)--(12,2)--(13,18)--(14,4)--(15,16)--(16,6)--(17,14)--(18,8)--(19,12)--(19.92,1);
\draw[dotted] (9,19)--(11.1,19) (10,1)--(19.92,1);
\draw (10,-3) node {\large  Stage 2};
\end{tikzpicture}}\hspace{1cm}
\scalebox{0.7}{\begin{tikzpicture}[line width=1pt,xscale=0.4,yscale=0.33,font=\small]
\draw[color=gray!30!white,step=2cm] (0,0) grid (20,20) ;
\draw \foreach \x in {1,3,...,21}{(\x-1,-0.8) node{\x}  }
 \foreach \x in {0,2,...,20}{ (-0.8,\x) node{\x} };
\draw (10,1)
--(11,20)--(12,2)--(13,18)--(14,4)--(15,16)--(16,6)--(17,14)--(18,8)--(19,12)--(19.92,1);
\draw[dotted]  (10,1)--(19.92,1);
\draw (10,-3) node {\large Stage 3};
\end{tikzpicture}}
}
\end{figure}

The range of stage 1 can be transformed to $\Lambda_{k-1}$ by minor changes of coordinates.
Clearly stages 2 and 3 each contribute $(k-1)$ edges to the path in $\RG(\Lambda_k)$.
Thus, if $T_k$ is the number of edges in the path from $p_1p_{2k+1}$ to $p_sp_s$ in $\RG(\Lambda_k)$,  we have $T_1=1$, and $T_k\ =\ T_{k-1}+2(k-1)$,
from which we can easily deduce $T_k=k(k-1)+1=(n-1)(n-2)/4 +1$.  The path from $p_1p_n$ to $p_sp_s$ is unique,
and has path length $(n-1)(n-2)/4$ in $\RG(\Lambda_k)$. Observe that this closely matches our upper bound $(n-1)^2/4$.

We can compute the number of vertices in $\RG(\Lambda_k)$ exactly. If $k$ is even (resp. odd), the total number of horizontal intersections from peaks to the left (resp. right) of $p_s$ with slopes to its right (resp. left) is
\[ (k-1)+(k-3)+(k-5)+\cdots+3+1\,.\]
The total number of horizontal intersections from valleys to the left (resp. right ) of $p_s$ with slopes to its right (resp. left) is
\[ (k-2)+(k-4)+(k-6)+\cdots+4+2\,.\]
These total $k(k-1)/2$, whether $k$ is odd or even.
The total number of horizontal intersections from peaks to the right (resp. left) of $p_s$ with slopes to its left (resp. right) is
\[ (k-1)+(k-3)+(k-5)+\cdots+3+1\,,\]
and the total number of horizontal intersections from valleys to the right (resp. left) of $p_s$ with slopes to its left (resp. right) is
\[ (k-2)+(k-4)+(k-6)+\cdots+4+2\,.\]
These also total $k(k-1)/2$. So the number of vertices in $\RG(\Lambda_k)$ is $k(k-1)+2$.
Hence $\RG(\Lambda_k)$ is a single path from $p_1p_n$ to $p_sp_s$.

To summarise, we have proved
\begin{theorem}\label{thm:lb}
Let $R(n)$ be the maximum possible number of events required for two parallel climbers to meet at the summit on a range $\Pi$ with $n$ nodes. Then $R(n)= n^2/4-O(n)$.
\end{theorem}

From this we may deduce
\setcounter{theorem}{17}
\begin{lemma}
Let $\Pi$ be as defined in Section A.1.  There are continuous functions $\focusA,\focusB:[0,1]\to\Pi$ satisfying
$\focusA(0)=p_1$, $\focusA(1)=q_k$, $\focusB(0)=q_n$, $\focusB(1)=p_{k+1}$, and $y(\focusA(t))=y(\focusB(t))$ for
all $t\in[0,1]$.  Moreover the set of events
\[ T=\big\{t\in[0,1]:\focusA(t)\in P\text{ or }\focusB(t)\in P\big\} \]
has cardinality at most $n^2$.
\end{lemma}
\begin{proof}
  $\Pi$ is a range with $2n$ nodes $\{p_1,q_1,p_2,q_2,\ldots,
p_n,q_n\}$. Hence, by Lemma A.1, $\RG(\Pi)$ has at most $(2n)^2/4=n^2$ vertices. Observe, from Theorem A.2, that this bound cannot be significantly improved in the worst case.
\end{proof}

\subsection*{Acknowledgements}
We thank Mary Cryan for useful discussions.

\end{document}